\documentclass[11pt, a4paper]{amsart}
\usepackage{a4wide}
\usepackage{mathrsfs, amssymb,amsmath,url,amsthm}

\title[Equations in oligomorphic clones and the CSP]{Equations in oligomorphic clones and the Constraint Satisfaction Problem for  $\omega$-categorical structures}
\date{\today}

\author[L.~Barto]
{Libor Barto}
	\address{Department of Algebra, MFF UK, Sokolovska 83, 186 00 Praha 8, Czech Republic}
	\email{libor.barto@gmail.com}
	\urladdr{http://www.karlin.mff.cuni.cz/~barto/}

\author[M.~Kompatscher]
{Michael Kompatscher}
    \address{Institut f\"{u}r Diskrete Mathematik und Geometrie, FG Algebra, TU Wien, Austria, and Department of Algebra, Charles University, Czech Republic}
\email{michael@logic.at}
\urladdr{https://www.logic.at/staff/kompatscher/}

\author[M.~Ol\v{s}\'{a}k]
{Miroslav Ol\v{s}\'{a}k}
	\address{Department of Algebra, MFF UK, Sokolovska 83, 186 00 Praha 8, Czech Republic}
	\email{mirek@olsak.net}

\author[T.~van Pham]
{Trung van Pham}
    \address{Institute of Mathematics, Mathematics for Computer Sciences, Hanoi, Vietnam}
    \email{pvtrung@math.ac.vn}
    \urladdr{https://www.logic.at/staff/pvtrung/}

\author[M.~Pinsker]
{Michael Pinsker}
	\address{Institut f\"{u}r Diskrete Mathematik und Geometrie, FG Algebra, TU Wien, Austria, and Department of Algebra, Charles University, Czech Republic}    
	\email{marula@gmx.at}
    \urladdr{http://dmg.tuwien.ac.at/pinsker/}
\thanks{Libor Barto and Miroslav Ol\v{s}\'{a}k were supported by the the Grant Agency of the Czech Republic, grant GA\v CR
13-01832S. The research of Michael Kompatscher, Trung Van Pham, and Michael Pinsker has been funded through project P27600 of the  Austrian Science Fund (FWF)}

\theoremstyle{plain}

    \newtheorem{thm}{Theorem}[section]
    \newtheorem{theorem}[thm]{Theorem}

    \newtheorem{lem}[thm]{Lemma}
    \newtheorem{lemma}[thm]{Lemma}

    \newtheorem{prop}[thm]{Proposition}
    \newtheorem{proposition}[thm]{Proposition}

    \newtheorem{cor}[thm]{Corollary}

    \newtheorem{conj}[thm]{Conjecture}

\theoremstyle{definition}

    \newtheorem{defn}[thm]{Definition}
    \newtheorem{definition}[thm]{Definition}

\theoremstyle{remark}

\def\mA{\mathbb{A}}

\def\mN{\mathbb{N}}
\def\csp{{\rm CSP}}
\def\pol{{\rm Pol}}
\def\mQ{{\mathbb{Q}}}
\def\mi{{\rm mi}}
\def\mx{{\rm mx}}
\def\lex{{\rm lex}}
\def\mP{{\rm \mathbb{P}}}

{
\theoremstyle{definition}
\newtheorem*{observation*}{Observation}

}

\def\mD{\mathbb{D}}

\DeclareMathOperator{\id}{id}
 
\DeclareMathOperator{\Aut}{Aut}

\DeclareMathOperator{\End}{End}
\DeclareMathOperator{\Pol}{Pol}

\DeclareMathOperator{\CSP}{CSP}

\newcommand{\ignore}[1]{}

\newcommand{\C}{{\mathscr C}}

\newcommand{\Q}{{\mathbb Q}}

\newcommand{\To}{\rightarrow}

\newcommand{\M}{{\mathscr M}}

\newcommand{\cM}{\tilde{\M}}
\newcommand{\cC}{\tilde{\C}}

\newcommand{\rest}{{\upharpoonright}}

\newcommand{\clone}[1]{\mathscr{#1}}
\newcommand{\relstr}[1]{\mathbb{#1}}

\newcommand{\cloA}{\clone{A}}
\newcommand{\cloB}{\clone{B}}
\newcommand{\cloC}{\clone{C}}
\newcommand{\relA}{\relstr{A}}
\newcommand{\relB}{\relstr{B}}
\newcommand{\relC}{\relstr{C}}
\newcommand{\relP}{\relstr{P}}
\newcommand{\relQ}{\relstr{Q}}
\newcommand{\relD}{\relstr{D}}

\newcommand{\relS}{\relstr{S}}

\begin{document}

\maketitle

\begin{abstract}
There exist two conjectures for constraint satisfaction problems (CSPs) of reducts of finitely bounded homogeneous structures: the first one states that tractability of the CSP of such a structure is, when the structure is a model-complete core, equivalent to its polymorphism clone satisfying a certain non-trivial linear identity modulo outer embeddings. The second conjecture, challenging the approach via model-complete cores by reflections, states that tractability is equivalent to the linear identities (without outer embeddings) satisfied by its polymorphisms clone, together with the natural uniformity on it, being non-trivial.

We prove that the identities satisfied in the polymorphism clone of a structure allow for conclusions about the orbit growth of its automorphism group, and apply this to show that the two conjectures are equivalent. We contrast this with a counterexample showing that $\omega$-categoricity alone 
is insufficient to imply the equivalence of the two conditions above in a model-complete core.

Taking a different approach, we then show how the Ramsey property of a homogeneous structure can be utilized for obtaining a similar equivalence under different conditions.

We then prove that any polymorphism of sufficiently large arity which is totally symmetric modulo outer embeddings of a finitely bounded structure can be turned into a non-trivial system of linear identities, and obtain non-trivial linear  identities for all tractable cases of reducts of the rational order, the random graph, and the random poset.

Finally, we provide a new and short proof, in the language of monoids, of the theorem stating that every $\omega$-categorical structure is homomorphically equivalent to a model-complete core.
\end{abstract}

\section{Introduction}\label{sect:intro}

In order to keep the presentation of the wide topic of the present article as compact as possible, we postpone most definitions to an own preliminaries section (Section~\ref{sect:prelims}).

\subsection{Constraint Satisfaction Problems}\label{subsect:CSP}

The Constraint Satisfaction Problem (CSP) of a structure $\relA$ in a finite relational language, denoted by $\CSP(\relA)$, is the computational problem of deciding its primitive positive theory: given a sentence $\phi$ which is an existentially quantified conjunction of atomic formulas, decide whether or not $\phi$ holds in $\relA$. When $\relA$ has a finite domain, then its CSP is in NP, and it has been conjectured that its CSP is always either NP-complete or polynomial-time solvable~\cite{FederVardi}. While the CSP of structures with an infinite domain can be of any complexity~\cite{BodirskyGrohe}, and can in particular be undecidable, for a certain class of infinite-domain CSPs a similar dichotomy conjecture as for the finite case has been stated. In fact, two such conjectures have been brought up via different approaches; in the present article we first establish their equivalence, and then investigate in more detail the 
tractability conditions of the two conjectures.

The range of both conjectures are reducts of finitely bounded homogeneous structures, a (proper) subclass of the countable $\omega$-categorical structures. It is well-known, and easy to see from the definition, that the CSP of such structures is contained in NP; both conjectures state that it is always either NP-complete or contained in P, but each conjecture gives a different delineation between the (NP-)hard and the tractable (i.e., polynomial-time solvable) cases.

\subsection{The first conjecture}\label{subsect:first}

The first and older conjecture, formulated by Bodirsky and Pinsker (cf.~\cite{BPP-projective-homomorphisms}), is based on the notion of the model-complete core of an $\omega$-categorical structure, which can be viewed as the simplest representative in the class of an $\omega$-categorical structure with respect to the equivalence relation of homomorphic equivalence. We have the following.
\begin{thm}[Bodirsky~\cite{Cores-journal}]\label{thm:cores}
Every countable $\omega$-categorical structure is homomorphically equivalent to a model-complete core. This model-complete core is unique up to isomorphism and itself $\omega$-categorical.
\end{thm}
The idea leading to the first conjecture is that the complexity of the CSP of a structure $\relA$ in the range of the conjecture is determined by which finite structures $\relC$ have a primitive positive (pp-) interpretation with parameters in its model-complete core $\relB$. This approach builds on two facts: the first fact being that homomorphically equivalent structures have the same primitive positive theory, and hence $\relA$ and $\relB$ have equal CSPs; and the second fact being that if a structure $\relC$ has a primitive positive interpretation with parameters in an $\omega$-categorical model-complete core $\relB$, then $\CSP(\relC)$ reduces to $\CSP(\relB)$ in polynomial time. It is a well-known fact that the structure 
$$\relS:=(\{0,1\};\{(0,0,1),(0,1,0),(1,0,0)\})$$ pp-interprets all finite structures, and that its CSP is NP-complete.

\begin{conj}\label{conj:old}
Let $\relA$ be a reduct of a finitely bounded homogeneous structure, and let $\relB$ be its model-complete core. Then one of the following holds.
\begin{itemize}
\item[(i)] $\relB$ pp-interprets $\relS$ with parameters (and consequently, $\CSP(\relA)$ is NP-complete).
\item[(ii)] $\CSP(\relA)$ is polynomial-time solvable.
\end{itemize}
\end{conj}
From our remarks above it follows that if condition (i) in Conjecture~\ref{conj:old} holds, then $\CSP(\relA)$ is indeed NP-complete. What remains to prove is that if this condition is not satisfied, then $\CSP(\relA)$ is tractable. The following equivalent conditions have been established for this situation via the polymorphism clone $\Pol(\relB)$ of a structure $\relB$ ((ii) in~\cite{Topo-Birk}, and (iii), (iv) in~\cite{BartoPinskerDichotomy}). We denote the clone of projections on the set $\{0,1\}$ by $\mathbf 1$; then $\mathbf 1=\Pol(\relS)$.

\begin{thm}\label{thm:siggers}
Let $\relB$ be an $\omega$-categorical model-complete core. The following are equivalent.
\begin{itemize}
\item[(i)] $\relB$ does not pp-interpret $\relS$ with parameters.
\item[(ii)] No stabilizer of $\Pol(\relB)$ maps homomorphically and continuously to $\mathbf 1$.
\item[(iii)] No stabilizer of $\Pol(\relB)$ maps homomorphically to $\mathbf 1$.
\item[(iv)] $\Pol(\relB)$ has a Siggers term modulo outer embeddings, i.e., there exist $e_1,e_2, f\in\Pol(\relB)$ such that the identity 
$$
e_1\circ f(x,y,x,z,y,z)= e_2\circ f(y,x,z,x,z,y)
$$
holds in $\Pol(\relB)$.
\end{itemize}
\end{thm}
Observe that the very recent condition~(iv) turns, for the first time, the supposed tractability criterion of Conjecture~\ref{conj:old} into a positive statement, nourishing the hope for a positive answer to the conjecture.

\subsection{The second conjecture}\label{subsect:second}

The second and younger conjecture was born from the observation that the usage of homomorphic equivalence and pp-interpretations might not be optimal in the order which leads to Conjecture~\ref{conj:old}, as the crucial structure $\relS$ might, for example, be homomorphically equivalent to a structure with a pp-interpretation in $\relA$, but not pp-interpretable with parameters by the model-complete core of $\relA$. This suggests the following weaker conjecture, which does use the reductions by homomorphic equivalence and pp-interpretations in the best possible way~\cite{wonderland}.

\begin{conj}\label{conj:new}
Let $\relA$ be a reduct of a finitely bounded homogeneous structure. Then one of the following holds.
\begin{itemize}
\item[(i)] $\relS$ is homomorphically equivalent to a structure with a pp-interpretation in $\relA$ (and consequently, $\CSP(\relA)$ is NP-complete).
\item[(ii)] $\CSP(\relA)$ is polynomial-time solvable.
\end{itemize}
\end{conj}

It has been remarked in~\cite{wonderland} that the two conjectures are equivalent for finite structures. While more likely to be true, one disadvantage of Conjecture~\ref{conj:new} is that there is no unique optimal structure that can be pp-interpreted in $\relA$, as opposed to the model-complete core for homomorphic equivalence. Similarly to (ii) in Theorem~\ref{thm:siggers}, the authors of~\cite{wonderland} did however provide an equivalent tractability criterion using identities and topology.

\begin{thm}\label{thm:h1}
Let $\relA$ be $\omega$-categorical. The following are equivalent.
\begin{itemize}
\item[(i)] $\relS$ is not homomorphically equivalent to a structure with a pp-interpretation in $\relA$.
\item[(ii)] $\Pol(\relA)$ does not have a uniformly continuous h1 clone homomorphism to $\mathbf 1$.
\end{itemize}
\end{thm}

Note that a positive statement equivalent to the statements of Theorem~\ref{thm:h1}, i.e., an analogue of (iv) in Theorem~\ref{thm:siggers} is missing, leaving Conjecture~\ref{conj:new} somewhat less accessible than Conjecture~\ref{conj:old}.

\subsection{Equivalence of the conjectures}\label{subsect:intro:equivalence}

The results known so far concerning identities in polymorphism clones were shown for all $\omega$-categorical model-complete cores, rather than for the considerably more restricted class of structures concerned by the conjectures; it is probably fair to say that it seemed inconceivable that assumptions like finite boundedness would be useful when proving such structural results (these assumptions are, however, essential for the algorithmic aspects of the CSPs). Therefore, the most likely way of showing the equivalence of the conjectures seemed by proving that for all $\omega$-categorical model-complete cores, conditions (iv) in Theorem~\ref{thm:siggers} and (ii) in Theorem~\ref{thm:h1} are equivalent: that is, since the other implication is well-known and easy, that a Siggers term modulo outer embeddings prevents a uniformly continuous h1 clone homomorphism to $\mathbf 1$. 

We will, however, provide a counterexample, basically the atomless Boolean algebra with the right choice of relations, showing that this is not true in general.

\begin{thm}\label{thm:counterexample}
There exists an $\omega$-categorical model-complete core structure $\relA$ such that:
\begin{itemize}
\item[(i)] No stabilizer of $\Pol(\mathbb A)$ has a continuous clone homomorphism to $\mathbf 1$ (and hence, by Theorem~\ref{thm:siggers},  $\Pol(\relA)$ has a Siggers term modulo outer embeddings).
\item[(ii)] $\Pol(\mathbb A)$ has a uniformly continuous h1 clone homomorphism to $\mathbf 1$.
\end{itemize}
\end{thm}

Surprisingly, on the other hand, it turns out that every structure $\relA$ which is a counterexample as above must have at least double exponential orbit growth. This is remarkable in that it is the first instance discovered where structural higher-arity information about the polymorphism clone of an $\omega$-categorical structure yields information about its automorphism group.

\begin{thm}\label{thm:equivalence}
Let $\relA$ be a structure with the properties stated in Theorem~\ref{thm:counterexample}. Then its automorphism group $\Aut(\relA)$ must have at least double exponential orbit growth.
\end{thm}

From this, it is straightforward to derive the equivalence of the CSP conjectures, answering Problem~8.3 in~\cite{wonderland} to the positive, and showing that the implication from (4) to (3) in Corollary 5.3 of~\cite{BartoPinskerDichotomy} holds.

\begin{cor}\label{cor-equivconjectures}
Let $\mathbb A$ be a reduct of a structure which is homogeneous in a finite relational language, and let $\mathbb B$ be its model-complete core. Then the following are equivalent.
\begin{itemize}
\item[(i)] Some stabilizer of $\Pol(\mathbb B)$ has a continuous clone homomorphism to $\mathbf 1$.
\item[(ii)] $\Pol(\mathbb A)$ has a uniformly continuous h1 clone homomorphism to $\mathbf 1$.
\item[(iii)] $\Pol(\mathbb B)$ has a uniformly continuous h1 clone homomorphism to $\mathbf 1$.
\end{itemize}
In particular, Conjecture~\ref{conj:old} holds if and only if Conjecture~\ref{conj:new} holds.
\end{cor}

\subsection{The Ramsey property}\label{subject:intro:Ramsey}

Via an alternative approach involving Ramsey theory, we will then show a statement similar to Corollary~\ref{cor-equivconjectures} under different, and incomparable, conditions. Although this might seem irrelevant for CSPs considering our results above which cover the entire range of Conjectures~\ref{conj:old} and~\ref{conj:new}, it is interesting that various conditions of very different nature seem to imply this statement, while at the same time we know from our counterexample in Theorem~\ref{thm:counterexample} that $\omega$-categoricity alone is not sufficient. Observe that in the following theorem, there is no requirement of finite language, or orbit growth, or even $\omega$-categoricity; on the other hand, we require the non-trivial linear identities to be satisfied modulo embeddings of an ordered Ramsey structure.

\begin{thm}\label{thm:Ramsey}
Let $\relA$ be a reduct of an ordered homogeneous Ramsey structure $\relD$. If  $\Pol(\relA)$ satisfies a non-trivial set of linear identities modulo outer embeddings of $\relD$, then $\Pol(\relA)$ does not have a uniformly continuous h1 clone homomorphism to $\mathbf 1$.
\end{thm}

Note that Theorem~\ref{thm:Ramsey} corresponds to the contrapositive of the non-trival implication from (iii) to (i) in Corollary~\ref{cor-equivconjectures}, via the fact that (i) there is equivalent to the existence of a Siggers term modulo outer embeddings.

We would also like to remark that the situation of Theorem~\ref{thm:Ramsey} is particularly interesting for the approach to Conjectures~\ref{conj:old} and~\ref{conj:new} via canonical functions, as surveyed in~\cite{BP-reductsRamsey} (cf.~also the recent~\cite{canonical}); indeed, many of the successful CSP classifications via that approach yield tractable situations as in Theorem~\ref{thm:Ramsey}.

\subsection{Linearization}\label{subsect:linearization}

Corollary~\ref{cor-equivconjectures}, combined with Theorem~\ref{thm:siggers}, implies that if an $\omega$-categorical model-complete core $\mathbb B$ has a Siggers polymorphism modulo outer embeddings, then $\Pol(\mathbb B)$ does not have a uniformly continuous h1 clone homomorphism onto $\mathbf 1$. 
 It does not imply that in that situation, $\Pol(\relB)$ satisfies non-trivial linear identities, i.e., that $\Pol(\relB)$ does not have an h1 clone homomorphism to $\mathbf 1$ disregarding the uniformity on $\Pol(\relB)$. The situation in Theorem~\ref{thm:Ramsey} is similar. It is hitherto unknown under which conditions non-trivial linear identities modulo outer embeddings imply non-trivial linear identities in a polymorphism clone; as of today, we cannot even refute the possibility that the existence of an h1 homomorphism to $\mathbf 1$ implies the existence of a uniformly continuous such homomorphism in general. This question, for $\omega$-categorical model-complete cores, corresponds to the implication from (6) to (4) in~\cite{BartoPinskerDichotomy}.

Approaching this problem, we are going to show that under the assumption of finite boundedness, and stronger identities than the Siggers identity modulo outer embeddings, we can derive the satisfaction of non-trivial linear identities in a polymorphism clone.

\begin{thm}\label{thm:finitelybounded}
Let $\relA$ be a reduct of a finitely bounded homogeneous structure $\relD$ which is given by a set of forbidden substructures all of which have size at most $k \geq 3$.  If $\Pol(\relA)$ contains a $k$-ary polymorphism $f$ which is totally symmetric modulo outer embeddings of $\relD$, i.e., for all permutations $\rho$ of $\{1,\ldots,k\}$ satisfies an identity of the form
$$
e_{1, \rho}\circ f(x_1,\ldots,x_k)=e_{2, \rho}\circ f(x_{\rho(1)},\ldots,x_{\rho(k)}),
$$
where $e_{1, \rho}, e_{2, \rho}\in\overline{\Aut(\relD)}$, then $\Pol(\relA)$ does not have an h1 clone homomorphism to $\mathbf 1$.
\end{thm}

%

From Theorem~\ref{thm:finitelybounded} and the classifications in~\cite{ecsps},~\cite{tcsps-journal},~\cite{BodPin-Schaefer-both}, and~\cite{posetCSP16} it follows directly that most reducts of the rationals, the random graph, and the random partial order with tractable CSPs have a polymorphism clone satisfying non-trivial linear identities. Using a similar proof technique for the remaining cases we obtain the following.

\begin{thm}\label{thm:randomgraph} \label{thm:linearisationtheorm}
Let $\relA$ be a reduct of one of the following structures:
\begin{itemize}
\item $(\mathbb N;=)$;
\item the order $(\mathbb Q;\leq )$ of the rational numbers;
\item the random partial order;
\item the random graph. 
\end{itemize}
Then $\Pol(\relA)$ has a uniformly continuous h1 clone homomorphism to $\mathbf 1$ if and only if it has an h1 clone homomorphism to $\mathbf 1$. When $\relA$ has a finite language, then its CSP is tractable if and only if $\Pol(\relA)$ satisfies a non-trivial set of linear identities.
\end{thm}

\subsection{Cores}\label{subsect:cores}

Theorem~\ref{thm:cores} above stating the existence and uniqueness of the model-complete core of an $\omega$-categorical structure is of central importance for Conjecture~\ref{conj:old}, and calculating the model-complete core of structures has been an integral part of the major successful CSP classifications so far. While the alternative more recent Conjecture~\ref{conj:new} threatened to make the notion obsolete for its context, the equivalence of the conjectures established in the present article provides further evidence of the decisive role of model-complete cores for CSPs.

Observe that the notion of a model-complete core is defined via the endomorphism monoid of a structure (density of the invertibles in the monoid), so in particular structures with isomorphic (as topological monoids, cf.~\cite{Reconstruction, BodirskyEvansKompatscherPinsker}) endomorphism monoids are either both model-complete cores, or none of them is. Moreover, by the theorem of Ryll-Nardzewski, Engeler, and Svenonius~\cite{Hodges}, the condition of $\omega$-categoricity of a countable structure is equivalent to oligomorphicity of its automorphism group, again captured by its endomorphism monoid. It thus seems natural to have a proof of Theorem~\ref{thm:cores} in the language of transformation monoids, without reference to the particular language of a structure. The original and quite lengthy proof due to Bodirsky, however, does work with structures, and it is not obvious how to translate it into a proof via monoids.

We shall provide a new, short proof of Theorem~\ref{thm:cores} using topological monoids, which perhaps reflects better the combinatorial content of the theorem, and in particular connects it to the recent notion of reflections (which in turn leads to the other conjecture, Conjecture~\ref{conj:new}). Set naturally in the language of monoids, our proof yields simultaneously the generalization of the theorem to weakly oligomorphic structures given in~\cite{PechCores}.

\subsection{Organization of this article}\label{subsect:organization}
We provide definitions and notation in Section~\ref{sect:prelims}. The main results about the two CSP conjectures, Theorems~\ref{thm:counterexample} and~\ref{thm:equivalence},  Corollary~\ref{cor-equivconjectures}, and Theorem~\ref{thm:Ramsey}, are shown in Section~\ref{sect:equiv}. In Section~\ref{sect:linearization}, we  investigate the relationship between linear identities modulo outer  embeddings and those without outer embeddings, proving Theorems~\ref{thm:finitelybounded}, and~\ref{thm:randomgraph}. The new proof of Theorem~\ref{thm:cores}, and new insights connecting it directly to reflections, are provided in Section~\ref{sect:cores}. 

 \section{Preliminaries} \label{sect:prelims}

We explain the notions which appeared in the introduction, and fix some notation for the rest of the article. 
For undefined universal algebraic concepts and more detailed presentations of the notions presented here we  refer to~\cite{BS81,Berg11}. For notions from model theory we refer to~\cite{Hodges}.

\subsection{Polymorphism clones, automorphisms, and invertibles} We denote relational structures by $\relA, \relB$, etc. When $\relA$ is a relational structure, we reserve the symbol $A$ for its domain. We write $\Pol(\relA)$ for its \emph{polymorphism clone}, i.e., the set of all finitary operations on $A$ which preserve all relations of $\relA$. The polymorphism clone $\Pol(\relA)$ is always a \emph{function clone}, i.e., it is closed under composition and contains all projections. The unary functions in  $\Pol(\relA)$ are precisely the \emph{endomorphisms} of $\relA$, denoted by $\End(\relA)$. The endomorphisms which are bijections and whose inverse function is also an endomorphism are precisely the \emph{automorphisms} of $\relA$. We denote the set of automorphisms of $\relA$ by $\Aut(\relA)$.

When $\cloA$ is any function clone (not necessarily the polymorphism clone of a structure), then still the unary functions in $\cloA$ form a transformation monoid, and the unary invertible functions in $\cloA$ (i.e., those having an inverse in $\cloA$) form a permutation group, the \emph{group of invertibles of $\cloA$}. We write $A$ for the domain of the function clone $\cloA$.

\subsection{Clone homomorphisms} A \emph{clone homomorphism} from a function clone $\cloA$ to a function clone $\cloB$ is a mapping $\xi\colon \cloA\To\cloB$ which
\begin{itemize}
\item preserves arities, i.e., it sends every function in $\cloA$ to a function of the same arity in~$\cloB$;
\item preserves each projection, i.e., it sends the $k$-ary projection onto the $i$-th coordinate in $\cloA$ to the same projection in $\cloB$, for all $1\leq i\leq k$;
\item preserves composition, i.e., $\xi(f(g_1,\ldots,g_n))=\xi(f)(\xi(g_1),\ldots,\xi(g_n))$ for all $n$-ary functions $f$ and all $m$-ary functions $g_1,\ldots,g_n$ in $\cloA$.
\end{itemize}
For all $1\leq i\leq k$ we denote the $k$-ary projection onto the $i$-th coordinate by $\pi^k_i$, in any function clone and irrespectively of the domain of that clone. This slight abuse of notation allows us, for example, to express the second item above by writing $\xi(\pi^k_i)=\pi^k_i$.

A mapping $\xi\colon \cloA\To\cloB$ is called an \emph{h1 clone homomorphism} if it preserves arities and composition with projections, i.e., $\xi(f(g_1, \ldots, g_n)) = \xi(f)(g_1, \ldots, g_n)$ for all $n$-ary functions $f$ in $\cloA$ and all $m$-ary projections $g_1, \dots, g_n$.
If, in addition, $\xi$ preserves each projection, then it is called a \emph{strong h1 clone homomorphism}.
Note that an h1 clone homomorphism to $\mathbf 1$ is automatically strong.

\subsection{Identities / Equations}
The clone homomorphisms are those mappings between clones preserving \emph{identities}, i.e., universally quantified equations between terms built from the functions in clones (with an appropriate language providing a symbol for every element of the clones). The h1 clone homomorphisms are those mappings between clones preserving identities of height one, and strong h1 clone homomorphisms preserve all identities of height at most one, also known as \emph{linear identities}, i.e., identities where no nesting of functions is allowed; cf.~\cite{wonderland}. A \emph{linear identity modulo outer unary functions} is a universally quantified equation of the form $e_1\circ s=e_2\circ t$, where $e_1, e_2$ are unary and $s,t$ are terms of height at most one.

When $\relA$ is a relational structure, then a \emph{linear identity of $\Pol(\relA)$ modulo outer endomorphisms (automorphisms, embeddings)} is an identity of the form $e_1\circ s=e_2\circ t$ which holds in $\Pol(\relA)$, where $e_1, e_2 \in\Pol(\relA)$ are endomorphisms (automorphisms, embeddings) of $\relA$, and $s,t$ terms over $\Pol(\relA)$ of height at most one. Similarly, we speak of linear identities modulo outer embeddings (automorphisms, embeddings) of $\relD$, where $\relD$ is some other structure, with the obvious meaning.

A set of identities is \emph{non-trivial} if it is unsatisfiable in the clone $\mathbf 1$. Therefore, a function clone satisfies a non-trivial set of identities if and only if it does not have a clone homomorphism to $\mathbf 1$; it satisfies a non-trivial set of linear identities if and only if it does not have an h1 clone homomorphism to $\mathbf 1$. It follows from the compactness theorem of first-order logic that these non-trivial sets of identities can be chosen to be finite.

\subsection{Stabilizers}

When $\cloA$ is a function clone, and $F\subseteq A$ is a finite subset of its domain, then the (pointwise) \emph{stabilizer} of $F$ in $\cloA$, denoted by $\cloA_{F}$, is the function clone of all $f\in\cloA$ satisfying $f(a,\ldots,a)=a$ for all $a\in F$. We emphasize that we always understand stabilizers to be pointwise, and  always of a finite set.

We remark that when $\relA$ is a relational structure and $F \subseteq A$ is finite, then the stabilizer of $F$ in $\Pol(\relA)$ is the polymorphism clone of the structure obtained by enriching $\relA$ by a unary singleton relation $\{a\}$ for every $a\in F$.

\subsection{Topology} Every function clone is naturally equipped with the topology of pointwise convergence: in this topology, a sequence $(f_i)_{i\in\omega}$ of $n$-ary functions converges to an $n$-ary function $f$ on the same domain if and only if for all $n$-tuples $\bar a$ of the domain the functions $f_i$ agree with $f$ on $\bar a$ for all but finitely many $i\in\omega$. Therefore, every function clone gives rise to an abstract \emph{topological clone} which reflects this topology as well as the composition structure of the clone~\cite{Reconstruction}. 

We always imagine function clones to carry the pointwise convergence topology, which is, in the case of a countable domain, in fact induced by a metric, and in general by a uniformity~\cite{Reconstruction, uniformbirkhoff, schneider}. Then a mapping $\xi\colon \cloA\To\cloB$, where $\cloA$ and $\cloB$ are function clones, is continuous if and only if for all $f\in\cloA$ and all finite sets $B'\subseteq B$ there exists a finite set $A'\subseteq A$ such that for all $g\in\cloA$ of the same arity as $f$, if $g$ agrees with $f$ on $A'$, then $\xi(g)$ agrees with $\xi(f)$ on $B'$. It is uniformly continuous if and only if for all $n\geq 1$ and all finite  $B'\subseteq B$ there exists a finite $A'\subseteq A$ such that whenever two $n$-ary functions $f, g\in\cloA$ agree on $A'$, then their images $\xi(f), \xi(g)$ agree on $B'$. Note that in the case of mappings $\xi\colon \cloA\To\mathbf 1$, uniform continuity means that for every $n\geq 1$ 
there exists a finite $A'\subseteq A$ such that $\xi(f)$ only depends on the restriction of $f$ to ${A'}$, for all $n$-ary $f\in\cloA$. When $\xi$ is an h1 clone homomorphism, then $A'$ can be chosen independently of $n$.

We remark that the polymorphism clones of relational structures are precisely the  function clones which are complete with respect to this uniformity (or, put differently, closed in the function clone of all functions of the domain). Function clones on a finite domain are discrete.

\subsection{Oligomorphicity, $\omega$-categoricity and orbit growth}\label{sect:notions_infinite} 
Recall that by the theorem of Ryll-Nardzewski, Engeler, and Svenonius, a countable relational structure $\relstr{A}$ is \emph{$\omega$-categorical} if and only if its automorphism group $\Aut(\relA)$ is \emph{oligomorphic}, i.e., for every $n \geq 1$, the natural componentwise action of $\Aut(\relstr{A})$ on $A^n$ has only finitely many orbits. In particular finite structures are always $\omega$-categorical. Every countable $\omega$-categorical structure $\relA$ thus induces a monotone function on the positive natural numbers which assigns to every $n\geq 1$ the number of orbits of $n$-tuples with respect to $\Aut(\relA)$; we call this function the \emph{orbit growth} of $\relA$ (or of $\Aut(\relA)$). There exist $\omega$-categorical structures of arbitrarily fast orbit growth.

Similarly, we say that a function clone is \emph{oligomorphic} if its group of unary invertibles is, and we can hence naturally speak of the orbit growth of an oligomorphic function clone.

\subsection{Homogeneity, finite boundedness, and the Ramsey property}
The $\omega$-categorical structures concerned by the conjectures above are reducts of finitely bounded homogeneous structures. Here, following~\cite{RandomReducts} and numerous subsequent authors, we define a \emph{reduct} of a relational structure $\relA$ to be a relational structure on the same domain all of whose relations have a first-order definition in $\relA$ without parameters.

 A relational structure is \emph{homogeneous} if every partial isomorphism between finite  substructures extends to an automorphism of the entire structure.  A countable relational structure $\relA$ is \emph{finitely bounded} if it has a finite signature, and there exists a finite set $F$ of finite structures in its signature such that $\relA$ contains precisely those structures as induced substructures which embed no member of $F$. We are going to call every such $F$ a set of \emph{forbidden substructures (with respect to $\relA$)}.

A relational structure $\relD$ is \emph{Ramsey} if for all finite induced substructures $\relP, \relQ$ of $\relD$ and all functions $\chi$ from the isomorphic copies of $\relP$ in $\relD$ to $\{0,1\}$ there exists an isomorphic copy of $\relQ$ in $\relD$ on which $\chi$ is constant. It is \emph{ordered} if it first-order defines (without parameters) a linear order on its domain. For more details about Ramsey structures in this context, we refer to the surveys~\cite{BP-reductsRamsey},~\cite{Bodirsky-HDR}.

\subsection{Homomorphic equivalence and model-complete cores}

When relational structures $\relstr{A}$ and $\relstr{B}$ have the same signature, then we say that $\relstr{A}$ and $\relstr{B}$ are \emph{homomorphically equivalent} if there exists a~homomorphism $\relstr{A} \to \relstr{B}$ and a~homomorphism $\relstr{B} \to \relstr{A}$. A relational structure $\relstr{B}$ is called a \emph{model-complete core} if $\Aut(\relB)$ is dense in $\End(\relB)$, i.e., for every endomorphism $e$ of $\relstr{B}$ and every finite subset $B'$ of $B$ there exists an automorphism of~$\relstr{B}$ which agrees with $e$ on $B'$. When $\relstr{B}$ is finite, then this means that every endomorphism is an automorphism, and $\relstr{B}$ is simply called a \emph{core}. 

Similarly, we call a function clone or a transformation monoid a \emph{model-complete core} if the group of its invertible functions is dense in its unary functions.

\subsection{CSPs} For a finite relational signature $\Sigma$ and a $\Sigma$-structure $\relstr{A}$, the \emph{constraint satisfaction problem} of $\relstr{A}$, or $\CSP(\relstr{A})$ for short, is the membership problem for the class
\begin{align*}
\{ \relstr{C} \mid \, &\relstr{C} \mbox{ is a finite $\Sigma$-structure and } \\ &\mbox{there exists a~homomorphism $\relstr{C} \to \relstr{A}$}\} \enspace.
\end{align*}
An alternative definition of $\CSP(\relstr{A})$ is via primitive positive (pp-) sentences. Recall that a \emph{pp-formula} over $\relstr{A}$ is a first order formula which only uses predicates from $\relstr{A}$, conjunction, equality, and existential quantification. $\CSP(\relstr{A})$ can equivalently be phrased as the membership problem of the set of  pp-sentences which are true in $\relstr{A}$. 


\section{Equivalence of the Conjectures,  and the Ramsey  Property}\label{sect:equiv}

This section is divided into three parts: we first prove Theorem~\ref{thm:equivalence} and Corollary~\ref{cor-equivconjectures} in Section~\ref{subsect:orbit}, and then provide the counterexample of Theorem~\ref{thm:counterexample} in Section~\ref{subsect:counterexample}. Finally, we turn to applications of the Ramsey property in Section~\ref{subsect:Ramsey}, proving Theorem~\ref{thm:Ramsey}.

\subsection{Orbit Growth and Equivalence of the conjectures}\label{subsect:orbit}

\begin{defn}
Let $\C$ be a function clone, and let $S\subseteq C$ be a subset of its domain with $|S|\geq 2$. Then a function $g\colon C\To S$ is a \emph{retractional witness for $\C$ with respect to $S$} if the restriction of $g\circ t$ to $S$ is a projection on $S$ for all $t\in\C$.

For an $n$-ary $t\in\C$, we call an index $1\leq i \leq n$ \emph{fundamental for $t$ with respect to $S$} if there exists a retractional witness $g$ for $\C$ with respect to $S$ such that $g\circ t\rest_{S^n}$ is the $i$-th $n$-ary projection $\pi^n_i$ on $S$.

The \emph{ambiguity degree of $t\in\C$ with respect to $S$} is the number of its fundamental indices with respect to $S$. The ambiguity degree of $\C$ is the supremum of the ambiguity degrees of its members with respect to sets $S\subseteq C$ of at least two elements:
\begin{align*}
\sup\;\{ d\in\omega \mid \, &\exists t\in\C,\, S\subseteq C\; (|S|\geq 2\; \wedge\; t \text{ has ambiguity } \\ &\text{degree }d\text{ with respect to } S)\}
\end{align*}

\end{defn}

\begin{lemma}\label{lem:doubleexp}
Let $\C$ be a function clone of infinite ambiguity degree. Then the componentwise action of the group of unary invertible functions in $\C$ on $C^n$ has at least $2^{2^n}-1$ orbits, for all $n\geq 1$.
\end{lemma}
\begin{proof}

Given $n\geq 1$, pick $t(x_1,\ldots,x_k)\in\C$ of ambiguity degree at least $2^n$ with respect to $S\subseteq C$ of at least two elements; by taking a subset, we may assume $|S|=2$. By identifying some variables of $t$ with variables corresponding to a fundamental index of $t$, we may assume that all indices of $t$ are fundamental, and that $k= 2^n$. For any non-empty subset $R$ of $S^n$, pick an $n$-tuple $q^R\in C^n$ of the form $t(q_1^R,\ldots,q_k^R)$ (applied componentwise), where every $q_i^R\in R$, and all tuples in $R$ appear as some $q_i^R$.

We claim that when $R\neq R'$, then $q^R$ and $q^{R'}$ lie in distinct orbits. To see this, suppose without loss of generality that $R\setminus R'\neq \emptyset$. Thus there exists some $q_i^R\notin R'$. Let $g\colon C\To S$ be the retractional witness such that $g\circ t\rest_{S^k}$ projects to the $i$-th coordinate. If $q^R=\alpha(q^{R'})$ for an invertible $\alpha\in\C$, then we would have $g(q^R)=g\circ \alpha(q^{R'})=g\circ \alpha\circ t(q^{R'}_1,\ldots,q^{R'}_k)$. Observe that $g\circ (\alpha\circ t)\rest_{S^k}$ is a projection since $\alpha\circ t\in\C$ and since $g$ is a retractional witness. Hence, $q_i^R=g(q^R)\in\{q_1^{R'},\ldots, q_k^{R'}\}$, a contradiction.
\end{proof}

\begin{lem}\label{lem:infiniteambiguity}
Let $\C$ be a function clone which is a model-complete core and which satisfies some non-trivial linear identity modulo outer unary functions. If $\C$ has a  retractional witness, then it has infinite ambiguity degree.
\end{lem}
\begin{proof}
For any $t\in\C$ of ambiguity degree $n\geq 1$, we find $t'\in\C$ of ambiguity degree $2n$. So let $t\in\C$ be given, and let $S\subseteq C$ be a 2-element set such that $t$ has $n$ fundamental indices with respect to $S$, witnessed by functions $g_1,\ldots,g_n\colon C\To S$. By identifying  variables we may assume that $t$ is $n$-ary. Renaming the variables, we may further assume that $g_i$ witnesses the index $i$, for $1\leq i\leq n$. Set $R:=t[S^n]$. Because $\C$ is a model-complete core, in  the stabilizer $\C_R$ a nontrivial identity which is linear modulo outer unary functions is satisfied, since linear identities modulo outer functions which hold in a model-complete core also hold in all of its stabilizers (this is easy to see and well-known, but we refer to~\cite{BartoPinskerDichotomy}). Let $s\in \C_R$ witness this, i.e., $s$ satisfies the nontrivial identity $u\circ s(y_1,\ldots,y_{m})=v\circ s(z_1,\ldots,z_{m})$, for variables $y_1,\ldots,y_m,z_1,\ldots,z_m$ which are not necessarily distinct. We claim that the $nm$-ary term 
$$s\ast t:=s(t(x_1^1,\ldots,x_n^1),\ldots,t(x_1^m,\ldots,x_n^m))$$ 
has the desired property. 

To see this, we first observe that $g_i\circ (s\ast t)\rest_{S^{nm}}$ is a projection, and in fact a projection onto a variable of the form $x_i^j$: inserting variables $x_1,\ldots,x_n$ into $s\ast t$, we obtain
\begin{align*}
g_i\circ &(s\ast t)(x_1,\ldots,x_n,\ldots,x_1,\ldots,x_n)\rest_{S^n} \\
&=g_i\circ s(t(x_1,\ldots,x_n),\ldots,t(x_1,\ldots,x_n))\rest_{S^n}\\
& =g_i\circ t(x_1,\ldots,x_n)\rest_{S^n}\\&=\pi^n_i(x_1,\ldots,x_n)\rest_{S^n}, 
\end{align*}
with the second equation holding since $R$ is stabilized by $s$. In particular, $s\ast t$ has ambiguity degree at least $n$, witnessed by $g_1,\ldots,g_n$. Note furthermore that for the same reason, the functions $u\circ (s\ast t)$ and $v\circ (s\ast t)$ have ambiguity degree at least $n$, projecting to a variable of the form $x_i^j$ when composed with $g_i$ from the left. This can be restated by saying that $s\ast t$ has ambiguity degree at least $n$, with fundamental indices corresponding to variables of the form $x_i^j$ witnessed by two witnesses $g_i\circ u$ and $g_i\circ v$; we now argue that the fundamental indices witnessed by $g_i\circ u$ and $g_i\circ v$ are distinct.

To this end, fix any $1\leq i\leq n$, and say that $g_i\circ u\circ (s\ast t)\rest_{S^{nm}}$ projects onto $x_i^j$, where $1\leq j\leq {m}$. Since the equation $u\circ s(y_1,\ldots,y_{m})=v\circ s(z_1,\ldots,z_{m})$ is non-trivial, we must have $y_j\neq z_j$. On the other hand, we must have that $y_j\in\{z_1,\ldots, z_m\}$, since $s$ obviously depends on its $j$-th variable. Thus $y_j=z_{\ell}$ for some $\ell\neq j$. This means that $g_i\circ v$ is a retractional witness such that $(g_i\circ v)\circ (s\ast t)\rest_{S^{nm}}$ projects onto the variable with index $x_i^\ell$, proving our claim. 

Summarizing, each fundamental index of $t$, witnessed by some $g_i$, has a corresponding  fundamental index of $s\ast t$, also witnessed by $g_i$, and this assignment is injective; and moreover, each fundamental index of $s\ast t$, witnessed by $g_i$, yields two fundamental indices of $s\ast t$, witnessed by $g_i \circ u$ and $g_i\circ v$, respectively.
\end{proof}

We thus obtain the following theorem, which shows, in particular, how equational properties of the polymorphism clone of a structure can have implications about its automorphism group. We first formulate it in terms of function clones, and then restate it in terms of structures.

\begin{theorem}\label{thm:doubleexp}
Let $\C$ be an oligomorphic function clone which is a model-complete core. Suppose that 
\begin{itemize}
\item[(i)] $\C$ satisfies a non-trivial linear identity modulo outer unary functions, and 
\item[(ii)] $\C$ has a uniformly continuous h1 clone homomorphism onto $\mathbf 1$.
\end{itemize} 
Then $\C$ has at least double exponential orbit growth.
\end{theorem}
\begin{proof}
By the results from~\cite{wonderland}, (ii) together with oligomorphicity implies that there exists $\ell\geq 1$ such that the componentwise action of $\C$ on $C^\ell$, which we denote by $\C^\ell$, has a retractional witness (in the terminology of~\cite{wonderland}, which we avoid to fully define here, the clone $\mathbf 1$ is an expansion of a reflection of a finite power of $\C$, which implies our formulation -- see also Section~\ref{sect:cores}). 
By Lemma~\ref{lem:infiniteambiguity}, $\C^\ell$ has infinite ambiguity degree, and so it has at least double exponential orbit growth by Lemma~\ref{lem:doubleexp}. Hence, also $\C$ has at least double exponential orbit growth.
\end{proof}

\begin{cor}\label{cor:structuresdoubleexp}
Let $\mathbb A$ be an $\omega$-categorical model-complete core, and suppose that $\Pol(\relA)$ satisfies (i) and (ii) of Theorem~\ref{thm:doubleexp}. Then $\relA$ has at least double exponential orbit growth.
\end{cor}

Note that this implies, in particular, Theorem~\ref{thm:equivalence}.
We obtain the following result in the language of clone homomorphisms, for reducts of homogeneous structures in a finite language.

\begin{cor}\label{cor:equiv-cores}
Let $\mathbb A$ be a reduct of a structure which is homogeneous in a finite relational language, and suppose $\relA$ is a model-complete core. Then the following are equivalent.
\begin{itemize}
\item[(i)] Some stabilizer of $\Pol(\mathbb A)$ has a continuous clone homomorphism to $\mathbf 1$.
\item[(ii)] $\Pol(\mathbb A)$ has a uniformly continuous h1 clone homomorphism to $\mathbf 1$.
\end{itemize}
\end{cor}
\begin{proof}
The implication from (i) to (ii) is a direct consequence of the results in~\cite{wonderland}. For the other direction, assume that (ii) holds.
By Theorem~\ref{thm:siggers}, (i) holds if and only if $\Pol(\mathbb A)$ has no Siggers term modulo outer unary functions. If that was not the case, then Corollary~\ref{cor:structuresdoubleexp} would imply that $\Aut(\mathbb A)$ has at least double exponential orbit growth, contradicting that $\mathbb A$ is a reduct of a structure which is homogeneous in a finite relational language (see~\cite{MacphersonSurvey}).
\end{proof}

Finally, we obtain the equivalence of the two CSP conjectures.

\begin{proof}[Proof of Corollary~\ref{cor-equivconjectures}]
By~\cite{wonderland}, $\Pol(\mathbb A)$ has a uniformly continuous h1 clone homomorphism onto $\mathbf 1$ if and only if $\Pol(\mathbb B)$ does, so~(ii) and~(iii) are equivalent.  Applying Corollary~\ref{cor:equiv-cores} to the model-complete core $\mathbb B$, and taking into account that $\Aut(\mathbb B)$ does not have faster orbit growth than $\Aut(\mathbb A)$ (for the latter, refer to the proof of the existence of the model-complete core in Section~\ref{sect:cores}), the equivalence with~(i) follows.
\end{proof}

\subsection{The counterexample}\label{subsect:counterexample} 

We now prove Theorem~\ref{thm:counterexample}. That is, we show that in Corollary~\ref{cor:equiv-cores}, it would not be sufficient to only require the structure $\relA$ to be an $\omega$-categorical model-complete core: the assumption of being a reduct of a homogeneous structure in finite language (or more precisely, as we can see from the proof, the assumption of less than double exponential orbit growth) is indeed needed.


Our counterexample is based on the 
\textit{countable atomless Boolean algebra}, i.e., the (up to isomorphism) unique countable Boolean algebra without atoms (see e.g.~\cite{Hodges}). This Boolean algebra can be described, more explicitly, as the Boolean algebra that is freely generated by a countable set of generators. Among other interesting model-theoretical properties, it is $\omega$-categorical and has double exponential orbit growth. In the following we occasionally view this structure as a relational structure $\mathbb B:= (B; \land, \lor, \neg, 0,1)$, where the relations are the graphs of the fundamental operations of the Boolean algebra (although we will sometimes use the same symbols for the operations of the Boolean algebra, without danger of confusion). The following two statements about $\mathbb B$ are essential for the construction of our counterexample.

\begin{lemma} \label{lem:counterexample} \
Let $\mathbb B= (B; \land, \lor, \neg, 0,1)$ be the countable atomless Boolean algebra. Then:
\begin{itemize}
\item[(i)] For every finite set $C \subseteq B$ there is a binary injective $f\in\Pol(\mathbb B)$ which stabilizes all elements of $C$ and which is symmetric modulo outer embeddings of $\relB$, i.e., the identity $e_1\circ f(x,y)= e_2\circ f(y,x)$ holds for some self-embeddings $e_1,e_2$ of $\mathbb B$.
\item[(ii)] $\Pol(\mathbb B)$ has a uniformly continuous h1 clone homomorphism onto $\mathbf 1$.
\end{itemize}
\end{lemma}

\begin{proof}
Let $\{c_1,c_2,\ldots \}$ be a countable set that freely generates $\mathbb B$. Since every element of $\mathbb B$ can be expressed as a term over $\mathbb B$ using finitely many generators, we can restrict ourselves in (i) to the stabilizers of sets of the form $C = \{c_1,\ldots,c_n\}$. The product algebra $\mathbb B\times \mathbb B$ is also a countable atomless Boolean algebra and thus isomorphic to $\mathbb B$. Moreover, it is freely generated by the pairs $(1,0)$ and $(c_i,c_i)$ for all $i \geq 1$: let $(a,b) \in B \times  B$, let $\phi,\psi$ be terms over $\mathbb B$ such that $a = \phi(c_1,\ldots,c_n)$ and $b = \psi(c_1,\ldots,c_n)$ in $\mathbb B$. Then we can represent the pair $(a,b)$ by 
\begin{align*}
(a,b) = &(\phi((c_1,c_1),\ldots,(c_n,c_n)) \land (1,0)) \lor \\ &(\psi((c_1,c_1),\ldots,(c_n,c_n)) \land \neg(1,0)).
\end{align*}
We now define $f\colon \mathbb B\times \mathbb B \to \mathbb B$ to be the unique homomorphism that extends the following map between the generating sets:
\begin{align*}
(1,0) &\mapsto c_{n+1}, \\
(c_i,c_i) &\mapsto c_i \text{ for all } 1\leq i \leq n,\\
(c_i,c_i) &\mapsto c_{i+2} \text{ for all } i > n.
\end{align*}

By definition $f$ is a polymorphism of $\mathbb B$ that stabilizes all the elements $c_i$ for $1\leq i \leq n$. Furthermore, since $f$ is induced by a bijection between the generating sets of free Boolean algebras, $f$ is an isomorphism between $\mathbb B \times \mathbb B$ and $\mathbb B$. It satisfies the equation $f(x,y) = e \circ f(y,x)$, where $e$ denotes the unique automorphism of $\mathbb B$ that maps $c_{n+1}$ to $\neg c_{n+1}$ and fixes all other generating elements, which concludes the proof of~(i).

In order to show (ii), let $F$ be an ultrafilter of $\mathbb B$. Then for every $f \in \Pol(\mathbb B)$ exactly one of the elements $a_1 := f(1,0,0,\ldots,0)$, $a_2 := f(0,1,0,\ldots,0)$, \ldots, $a_n := f(0,0,0,\ldots,1)$ is an element of $F$: this follows from the fact that, since $f$ stabilizes $1$, the disjunction $a_1 \lor \cdots \lor a_n$ is equal to $1$; but on the other hand $a_i \land a_j$ is equal to $0$ whenever $i \neq j$, since $f$ stabilizes $0$. Let $i_f$ be the unique index such that $a_{i_f} \in F$. Then $\xi(f) := \pi_{i_f}^n\in\mathbf 1$  defines an  h1 clone homomorphism from $\Pol(\relB)$ to $\mathbf 1$. Furthermore $\xi$ is uniformly continuous, since for every $n\geq 1$ the image $\xi(f)$ of an $n$-ary polymorphism $f$ only depends on the restriction of $f$ to the finite set $\{0,1\}^n$.
\end{proof}

Note that the countable atomless Boolean algebra $\mathbb B$ is not a model-complete core, since it can be homomorphically mapped to the two-element Boolean algebra. However, a slight change of language yields a model-complete core which satisfies all conditions of Theorem~\ref{thm:counterexample}:

\begin{proof}[Proof of Theorem~\ref{thm:counterexample}]
Let $\mathbb A$ be the expansion of the countable Boolean algebra $\mathbb B$ by the inequality relation. Then clearly $\mathbb A$ and $\mathbb B$ have the same automorphism group. Using the fact that $\mathbb A$ contains the inequality relation, it can be easily verified that $\Aut(\mathbb A)$ is dense in the endomorphisms of $\mathbb A$, and so $\mathbb A$ is an $\omega$-categorical model-complete core.

By Lemma~\ref{lem:counterexample} (i), every stabilizer of $\Pol(\mathbb B)$ contains an injective binary function which is symmetric modulo outer embeddings. Since those functions are injective, they preserve in particular the inequality relation and are thus elements of $\Pol(\mathbb A)$. Therefore no stabilizer of $\Pol(\mathbb A)$ has a clone homomorphism to $\mathbf 1$. But, by Lemma~\ref{lem:counterexample} (ii) there is a uniformly continuous h1 clone homomorphism of $\Pol(\mathbb B)$ to $\mathbf 1$, and its restriction to $\Pol(\mathbb A)$ shows that also $\Pol(\mathbb A)$ has such a clone homomorphism. \end{proof}

\subsection{The Ramsey property}\label{subsect:Ramsey}
We now prove Theorem~\ref{thm:Ramsey}, which states that also under different, Ramsey-theoretic conditions, the satisfaction of a non-trivial set of linear identities modulo outer embeddings in a polymorphism clone implies that this clone has no uniformly continuous h1 clone homomorphism to $\mathbf 1$. Although the cases covered by this result are not congruent with the range of Conjecture~\ref{conj:old}, they appear in many known classifications of CSPs over homogeneous structures; in fact such CSP classifications are often based on the fact that the underlying structures can be expanded to Ramsey structures (cf.~\cite{BP-reductsRamsey, canonical} for numerous examples and further references).

Let $\relA$ be a reduct of an ordered homogeneous Ramsey structure $\relD$ and let $\Pol(\relA)$ satisfy a non-trivial set of linear identities modulo outer embeddings of $\relD$; by homogeneity, those embeddings are elements of $\overline{\Aut(\relD)}$. Then Theorem~\ref{thm:Ramsey} claims that there is no uniformly continuous h1 clone homomorphism from $\Pol(\relA)$ to $\mathbf 1$.  We provide two proofs, a combinatorial one applying the Ramsey property directly, and a more algebraic one using dynamical systems.

\begin{proof}[First proof of Theorem~\ref{thm:Ramsey}]
Let $\Pol(\mathbb A)$ satisfy the non-trivial set of identities 
$$ u^i \circ s^i(y_1^i,\ldots,y_{m}^i)=v^i \circ t^i(z_1^i,\ldots,z_{m}^i), $$
where $u^i,v^i \in \overline{\Aut(\mathbb D)}$, $s^i, t^i\in\Pol(\relA)$, and $y_j^i,z_j^i$ are not necessarily distinct variables, for $1 \leq i \leq k$ and $1 \leq j \leq m$. The existence of a finite non-trivial set of identities follows from the compactness theorem of first-order logic, and we can assume $s^1,\ldots,s^k,t^1,\ldots,t^k$ to have equal arity $m$ by adding dummy variables. Moreover assume, for technical reasons, that every right side of an identity also appears as a left side, simply by repeating identities. 
For contradiction, let us assume that there is a uniformly continuous h1 clone homomorphism $\xi\colon \Pol(\mathbb A) \to \mathbf 1$.

By the uniform continuity of $\xi$ there is a finite $F\subseteq A$ such that whenever two functions $f,g \in \Pol(\mathbb A)$ of arity $m$ agree on $F$, then $\xi(f)=\xi(g)$. Let $\relC^1,\ldots,\relC^k$ be the structures induced  in $\mathbb D$ by the sets $s^1[F^m],\ldots, s^k[F^m]$, respectively. We are going to color the copies of $\relC^1,\ldots,\relC^k$ in $\mathbb D$.

By the homogeneity of $\mathbb D$ all such copies have domains of the form $\alpha[s^i[F^m]]$, where $\alpha \in \Aut(\mathbb D)$. Since $\mathbb D$ is totally ordered, every other $\beta\in\Aut(\relD)$ that maps $s^i[F^m]$ to $\alpha[s^i[F^m]]$ has to coincide with $\alpha$ on $s^i[F^m]$. Hence, since $\xi(\alpha \circ s^i)$ only depends on the restriction of $\alpha \circ s^i$ to $F^m$, the coloring $\chi^i$ on the copies of $\relC^i$ which sends every copy induced by $\alpha[s^i[ F]]$ to $\xi(\alpha\circ s^i)$ is well defined.
 
Now set $\mathbb S$ to be the structure induced by the union of all the sets $s^i[F^m]$ and $u^i[s^i[F^m]]$, where $1 \leq i \leq k$. By the Ramsey property, there is an isomorphic copy $\mathbb S'$ of $\mathbb S$ in $\mathbb D$ on which the colorings $\chi^i$ are monochromatic. This implies that for any $\beta\in\Aut(\relD)$ that maps $\mathbb S$ to $\mathbb S'$ we have $\xi(\beta \circ u^i \circ s^i) = \xi(\beta \circ s^i)$ and $\xi(\beta \circ v^i \circ t^i) = \xi(\beta \circ t^i)$ for all $1 \leq i \leq k$. Hence, because $\xi$ preserves linear identities, 
\begin{align*}
\xi(\beta\circ s^i(y_1^i,\ldots,y_{m}^i))&= \xi(\beta\circ u^i\circ s^i(y_1^i,\ldots,y_{m}^i))\\
&=\xi(\beta\circ v^i\circ t^i(z_1^i,\ldots,z_{m}^i))\\
&=\xi(\beta\circ  t^i(z_1^i,\ldots,z_{m}^i))\; ,
\end{align*}
 contradicting the fact that the system of the identities $ \beta\circ s^i(y_1^i,\ldots,y_{m}^i)= \beta\circ t^i(z_1^i,\ldots,z_{m}^i)$ is unsatisfiable in $\mathbf 1$.
\end{proof}

\begin{proof}[Second proof of Theorem~\ref{thm:Ramsey}]
We will use the fact due to~\cite{Topo-Dynamics} that $\Aut(\mathbb D)$ is, as the automorphism group of an ordered Ramsey structure, \emph{extremely amenable}: whenever it acts continuously on a compact Hausdorff space, then this action has a fixed point.


Fix $m\geq 1$, and let $S_m$ be the set of all mappings from the $m$-ary functions in $\Pol(\relA)$ to the $m$-ary functions in $\mathbf 1$. 
Bearing the product topology, $S_m$ is a compact Hausdorff space.
We define an action of $\Aut(\relD)$ on $S_m$ 
by setting, for $\alpha\in\Aut(\relD)$ and $\xi\in S_m$, the mapping $(\alpha \cdot \xi) \in S_m$ to be given by
$$ (\alpha \cdot \xi) (f) := \xi (\alpha^{-1}\circ f) \text{ for all } m\text{-ary } f \in \Pol(\mathbb A).$$ 

For contradiction, suppose that there is a uniformly continuous h1 clone homomorphism from $\Pol(\mathbb A)$ to $\mathbf 1$, and let $\xi\in S_m$ be its restriction to $m$-ary functions, where $m\geq 1$ is fixed. 
Consider the restriction of the above action of $\Aut(\relD)$ to the closure of the orbit of $\xi$ in $S_m$, i.e., 
let $\Aut(\relD)$ act on 
$$
C:=\overline{\{\alpha \cdot \xi \mid \alpha\in \Aut(\relD)\}}\; .
$$
Clearly, $C$ is compact. Moreover, while the action of $\Aut(\relD)$ need not be continuous on $S_m$, its restriction to $C$ is. To illustrate this, let us first observe that there exists a finite set $F \subseteq A^m$ such that for all $\psi\in C$ and all  $m$-ary $f,f'\in\Pol(\relA)$ we have that $f\rest_F = f'\rest_F$ implies $\psi(f) = \psi(f')$. Now consider a basic open neighborhood 
$$
O_f(\psi)= \{ \psi' \in C \mid \psi'(f) = \psi(f) \}
$$ of some $\psi \in C$, where the $m$-ary $f \in \Pol(\mathbb A)$ is fixed. Then by our remark above, the set 
$$
\{(\alpha,\psi') \in \Aut(\relD) \times C\mid \alpha\text{ stabilizes } f[F], \text{ and } \psi'(f) = \psi(f) \}
$$ 
is a basic open neighborhood of $(\id,\psi)$ that is mapped into $O_f(\psi)$ under the action.

Since $\Aut(\relD)$ is extremely amenable, there is a fixed point $\xi'$ of its action on $C$, i.e., $(\alpha \cdot \xi') = \xi'$ for all $\alpha \in \Aut(\relD)$. 
This means that $\xi'$ preserves composition with elements of $\Aut(\relD)$ from the outside, and by continuity even with elements of $\overline{\Aut(\relD)}$, i.e., with self-embeddings of $\relD$. Moreover, $\xi'$ preserves linear identities, since any mapping $\alpha\cdot \xi$ does, and so does any mapping in the closure of the functions of the latter form. 

It follows that $\Pol(\relA)$ cannot satisfy any finite non-trivial set of identities which are linear modulo embeddings of $\relD$ from the outside, as otherwise they would be satisfied in $\mathbf 1$ by virtue of $\xi'$, if we choose $m$ larger than all arities of the functions in that set.
\end{proof}

We would like to remark that the atomless Boolean algebra $\mathbb B$ that was used to provide the counterexample of Theorem \ref{thm:counterexample} is the reduct of a homogeneous Ramsey structure, namely of its expansion by a linear order which extends the natural partial order on $\mathbb B$ (see for instance~\cite{Topo-Dynamics}). We proved in Lemma \ref{lem:counterexample} that there are polymorphisms of $\mathbb B$ satisfying the non-trivial equation $f(x,y) = e \circ f(y,x)$. However this non-trivial equation does not satisfy the condition of Theorem \ref{thm:Ramsey}, since the embedding $e$ does not preserve any linear order on the domain of $\mathbb B$.

  \section{Linearization of Non-trivial Identities}\label{sect:linearization}


We are going to show that, under stronger conditions than the existence of a Siggers term modulo outer embeddings, we can derive the satisfaction of non-trivial linear identities in polymorphism clones. 
In Section~\ref{subsect:totally} we prove a strengthening of Theorem \ref{thm:finitelybounded}. We then show  in Section~\ref{subsect:examples} how to apply this result and similar methods to the polymorphism clones of all reducts of equality, the rational order, the random graph and the random partial order, for which complete complexity classifications of the corresponding CSPs have been obtained~\cite{ecsps, tcsps-journal, BodPin-Schaefer-both, posetCSP16}.

\subsection{Totally symmetric polymorphisms modulo embeddings}\label{subsect:totally} 
The mentioned strengthening of Theorem \ref{thm:finitelybounded},  Proposition~\ref{prop:finitelybounded}, uses a weaker notion of total symmetry.

\begin{definition}\label{def:numinors}
Let $f(x_1,\ldots,x_k)$ be a $k$-ary operation on a set $D$. We define the \textit{nu-minors of $f$} as the binary functions $h_i^f(x,y) := f(x,\ldots,x,y,x,\ldots,x)$, where $1\leq i\leq k$ and the only $y$ is located on the $i$-th coordinate. When $\relD$ is a relational structure on $D$, we say that the nu-minors of $f$ are \textit{totally symmetric modulo outer embeddings of $\relD$} if for all permutations $\rho$ of $\{1,\ldots,k\}$ there are embeddings $e_\rho, e'_\rho$ of $\relD$ such that
$$
e_\rho\circ h_i^f(x,y) =e'_\rho\circ h_{\rho(i)}^f(x,y) \text{ for all } 1\leq i\leq k\;.
$$
\end{definition}
Clearly, whenever $f$ is totally symmetric modulo outer embeddings of $\relD$, then also its nu-minors are totally symmetric modulo outer embeddings of $\relD$. On the other hand, we remark that if $f$ is a \emph{weak near unanimity} function modulo outer embeddings of $\relD$, i.e., satisfies the identities
$$
e_1\circ h_1^f(x,y)=\cdots=e_k\circ h_k^f(x,y)
$$
for embeddings $e_1,\ldots,e_k$ of $\relD$, then this does not imply in an obvious way that its nu-minors are symmetric modulo outer embeddings. We will show the following.

\begin{proposition} \label{prop:finitelybounded}
Let $\relA$ be a reduct of a finitely bounded homogeneous structure $\relD$ whose age is given by a finite set of forbidden substructures all of which have size at most $k \geq 3$. 

If $\Pol(\relA)$ contains a $k$-ary function whose nu-minors are totally symmetric modulo outer embeddings of $\relD$,  then $\Pol(\relA)$ does not have an h1 clone homomorphism to $\mathbf 1$.
\end{proposition}

The proof of Proposition \ref{prop:finitelybounded} is based on the following easy observation, which relies on the pigeonhole principle. 
\begin{lemma} \label{lem:lineq}
Let $\cloC$ be a function clone, let $k \geq 2$, and assume that there are binary functions $g_i \in \cloC$ for all $i \in \{1,\ldots, 2k-1\}$ such that
for every injective $\psi\colon \{1,\ldots,k\}\To\{1,\ldots, 2k-1\}$ there exists $f_{\psi}(x_1,\ldots,x_k) \in \cloC$ whose nu-minors equal the functions $g_{\psi(1)},\ldots,g_{\psi(k)}$.
Then there is no h1 clone homomorphism from $\cloC$ to $\mathbf 1$.
\end{lemma}

\begin{proof}
If there was an h1 clone homomorphism from $\cloC$ onto $\mathbf 1$, then by the pigeonhole-principle there would be $\psi\colon \{1,\ldots,k\}\To\{1,\ldots, 2k-1\}$ such that $g_{\psi(1)}, \ldots, g_{\psi(k)}$ all are sent to the same projection. But this contradicts the fact that $g_{\psi(1)}, \ldots, g_{\psi(k)}$ are the nu-minors of $f_{\psi}$.
\end{proof}

It is further enough to find polymorphisms that satisfy the equations in Lemma \ref{lem:lineq} locally, by the following lemma which can be proven by a simple compactness argument.

\begin{lemma}[Lemma 3 in \cite{canonical}] \label{lem:endoexistence}
Let $\mD$ be an $\omega$-categorical structure. Let $J$ be a set, and for every $j \in J$, let $f_j$ and $g_j$ be functions on $D$ of the same arity $m\geq 1$ such that for every finite $F\subseteq A^m$ there is $\alpha_{j} \in \Aut(\mD)$ with $\alpha_{j} \circ f_j\rest_F = g_j\rest_F$. Then there are $(e_{j})_{j \in J}, e \in \overline{\Aut(\mD)}$ such that $e_j \circ f_j = e \circ g_j$ for all $j \in J$. Moreover, if for $j_1,j_2 \in J$ we have $\alpha_{j_1} = \alpha_{j_2}$ for every finite set $F$, then $e_{j_1} = e_{j_2}$.
\end{lemma}

We are going to construct the functions $g_i$ needed for Lemma \ref{lem:lineq} as suitable compositions of the nu-minors of $f$ with embeddings of $\mD$.

\begin{lemma}\label{lem:funcsconstruction}
Let $\mA$, $\mD$, and $f(x_1,\ldots,x_k)$ be as in Proposition \ref{prop:finitelybounded}, and let $F\subseteq A$ be finite. Then there are binary $g^F_1, \ldots, g^F_{2k-1}\in\Pol(\relA)$ such that for every injective $\psi\colon\{1,\ldots,k\} \To \{1,\ldots, 2k-1\}$ there exists $\alpha_\psi \in \Aut(\mD)$ such that $  g^F_{\psi(i)}\rest_{F^2} = \alpha_\psi\circ h_{i}^f\rest_{F^2}$ for all $1\leq i\leq k$.
\end{lemma}

\begin{proof}

Whenever $\psi\colon \{1,\ldots,k\}\To \{1,\ldots,2k-1\}$ is injective, we define a mapping $\varphi_{\psi}$ 
\begin{align*}
 F^2\times \psi[\{1,\ldots,k\}]\;&\To\; D\\
(a,b,i)\;&\mapsto\; h_{\psi^{-1}(i)}^f(a,b)\; .
\end{align*} 
Writing $\sim_{\psi}$ for the kernel of $\varphi_{\psi}$, we then naturally obtain a structure $\mathbb X_\psi$  in the language of $\relD$ on the set $(F^2\times \psi[\{1,\ldots,k\}])/\sim_{\psi}$ of kernel classes of $\psi$, in which we choose the relations to be so that the mapping from $\mathbb X_\psi$ to $\relD$ induced by $\varphi_{\psi}$ is an embedding.

The main point of our construction is the observation that because the nu-minors of $f$ are totally symmetric modulo outer embeddings of $\relD$, any two structures $\mathbb X_{\psi_1}, \mathbb X_{\psi_2}$ are isomorphic via the mapping that sends any kernel class $[(a,b,i)]_{\sim_{\psi_1}}$ to $[(a,b,\psi_2\circ \psi_1^{-1}(i))]_{\sim_{\psi_2}}$. For example, to see that this mapping is well-defined, note that by definition $(a,b,i)\sim_{\psi_1} (c,d,j)$ if and only if $h^f_{\psi_1^{-1}(i)}(a,b)=h^f_{\psi_1^{-1}(j)}(c,d)$; but this is the case, by definition, if and only if $(a,b,\psi_2\circ \psi_1^{-1}(i))\sim_{\psi_2} (c,d,\psi_2\circ \psi_1^{-1}(j))$. Similarly, one checks that the mapping is an isomorphism.

We define a binary relation $\sim$ on $F^2\times \{1,\ldots,2k-1\}$ by setting $(a,b,i) \sim (c,d,j)$ if and only if there is a $\psi$ such that $(a,b,i) \sim_\psi (c,d,j)$, and claim that it is transitive, and thus an equivalence relation. To see transitivity, note that by the total symmetry of nu-minors, $(a,b,i) \sim (c,d,j)$ is equivalent to the statement that for every $\psi$ containing $i$ and $j$ in its image $(a,b,i) \sim_\psi (c,d,j)$ holds. Now let $(a,b,i), (c,d,j), (u,v,m) \in F^2\times\{1,\ldots,2k-1\}$ with $(a,b,i) \sim (c,d,j)$ and $(c,d,j) \sim (u,v,m)$. Since $k \geq 3$, there is an injection $\psi$ such that $(a,b,i) \sim_\psi (c,d,j) \sim_\psi (u,v,m)$, and hence $(a,b,i) \sim _\psi (u,v,m)$.

Since the relations of the structures $\mathbb X_\psi$ agree on their intersections, we obtain a structure $\mathbb X$ on the equivalence classes of $\sim$, defined as the ``union'' of the structures $\mathbb X_\psi$. This structure $\mathbb X$ does not contain any forbidden substructures of $\relD$, since any $k$-element substructure of $\mathbb X$ is already contained in some structure $\mathbb X_\psi$, which in turn embeds into $\relD$. Hence, $\mathbb X$ embeds into $\relD$ via an embedding $\varphi$. For $1\leq i\leq 2k-1$ and $(a,b)\in F^2$, we now set $g_i^F(a,b):=\varphi([(a,b,i)]_{\sim})$.

Given $\psi\colon\{1,\ldots,k\} \To \{1,\ldots, 2k-1\}$ as in the lemma, it is clear from the definition of $\mathbb X_\psi$ that the tuples $(g_{\psi(i)}^F(a,b)\mid 1\leq i\leq k,\; (a,b)\in F^2)$ and $(h_i^f(a,b)\mid 1\leq i\leq k,\; (a,b)\in F^2)$ satisfy the same relations in $\relD$. By the homogeneity of $\relD$, the latter can be sent to the first via an automorphism $\alpha_\psi$ of $\relD$, which is what we had to show.
\end{proof}

We have now all the tools ready to prove Proposition \ref{prop:finitelybounded}.

\begin{proof}[Proof of Proposition \ref{prop:finitelybounded}]
Let $f(x_1,\ldots,x_k)\in\Pol(\relA)$ have totally symmetric nu-minors. For every finite $F\subseteq A$, fix functions $g_1^F,\ldots,g_{2k-1}^F$ provided by Lemma~\ref{lem:funcsconstruction}. By a compactness argument similar to the one proving Lemma~\ref{lem:endoexistence}, we can assume that these functions are independent of $F$, i.e., that there exist functions $g_1,\ldots,g_{2k-1}$ which have the property stated in Lemma~\ref{lem:funcsconstruction} for every finite $F\subseteq A$. More precisely, this is achieved by observing that when $g_1^F,\ldots,g_{2k-1}^F$ are replaced by $\beta^F \circ g_1^F,\ldots,\beta^F\circ g_{2k-1}^F$ for some $\beta^F \in\Aut(\mathbb D)$, then they retain the property of Lemma~\ref{lem:funcsconstruction}; and picking finite sets $F_j\subseteq A$ for $j\in\omega$ such that $\bigcup_{j\in\omega} F_j=A$, we can use the $\omega$-categoricity of $\mathbb D$ to choose $\beta^{F_j}\in\Aut(\mathbb D)$ in such a way that for each $1\leq i\leq 2k-1$ the sequence $(\beta^{F_j}\circ g_i^{F_j})_{j\in\omega}$ converges to a function $g_i$. The functions $g_1,\ldots,g_{2k-1}$ then have the property claimed above.

By Lemma \ref{lem:endoexistence} we obtain embeddings $e, e_\psi \in \overline{\Aut(\mD)}$ such that $e\circ g_{\psi(i)}= e_\psi\circ h_i^f$ for all $1\leq i\leq k$. Then the functions $f_\psi := e_\psi \circ f$ and their nu-minors $e\circ g_{\psi(1)},\ldots, e\circ g_{\psi(k)}$ satisfy the conditions of Lemma \ref{lem:lineq}, which concludes the proof.
\end{proof}

Observing that the assumption $k\geq 3$ was only needed to ``amalgamate" the kernels in the proof of Lemma~\ref{lem:funcsconstruction}, we obtain the following variant of Proposition~\ref{prop:finitelybounded} in which we trade the condition $k\geq 3$ for injectivity.

\begin{cor} \label{prop:finitelybounded2}
Let $\relA$ be a reduct of a finitely bounded homogeneous structure $\relD$ whose age is given by a finite set of forbidden substructures all of which have size at most $k \geq 2$. 

If $\Pol(\relA)$ contains a $k$-ary function which is injective and whose nu-minors are totally symmetric modulo outer embeddings of $\relD$,  then $\Pol(\relA)$ does not have an h1 clone homomorphism to $\mathbf 1$.
\end{cor}

\subsection{Examples of linearization}\label{subsect:examples}

We are now going to prove Theorem~\ref{thm:randomgraph}. That is, we are going to show that for any reduct $\mathbb A$ of equality, the order of the rationals, the random partial order, or the random graph, $\Pol(\relA)$ has no uniformly continuous h1 clone homomorphism to $\mathbf 1$ if and only if it satisfies a non-trivial set of linear identities. To this end, we are going to analyse the linear identities modulo embeddings obtained in the corresponding CSP classifications~\cite{ecsps, tcsps-journal, posetCSP16, BodPin-Schaefer-both}. In most of the cases, Theorem~\ref{thm:finitelybounded} provides us directly with the desired linear identities, but we do have to consider some cases separately. We present the proof for $(\mN;=)$ in Proposition~ \ref{lem:equalitylanguage}, for the order of the rationals in Proposition~\ref{lem:Qlinearisation}, for the random partial order in Proposition~\ref{lem:posetlinearisation}, and for the random graph in Proposition~\ref{lem:randomgraphlinearisation}.

%

\subsubsection{Reducts of equality}
For the reducts of $(\mN;=)$, the CSP classification in~\cite{ecsps} shows the following.
\begin{theorem}\label{thm:equlancomclas}
Let $\mA$ be a reduct of $(\mN;=)$. Then either
\begin{enumerate}
\item $\Pol(\mA)$ contains a constant or a binary injective function, or
\item there is a continuous clone homomorphism from $\Pol(\mA)$ to $\mathbf{1}$.
\end{enumerate}
If $\mA$ has a finite relational language, then $\csp(\mA)$ is tractable in the first case, and NP-complete in the second case.
\end{theorem}
Theorem~\ref{thm:finitelybounded} then yields Theorem~\ref{thm:linearisationtheorm} for such reducts.

\begin{prop}\label{lem:equalitylanguage}
Theorem \ref{thm:linearisationtheorm} holds for reducts of $(\mN;=)$.
\end{prop}

\begin{proof}
If a reduct $\mA$ has a constant polymorphism, then it has a binary such polymorphism $c$, which clearly satisfies the non-trivial linear identity $c(x,y) = c(y,x)$. If $\Pol(\mA)$ contains a binary injection $f(x,y)$, then $f(x,f(y,z))$ is an injective ternary polymorphism of $\mA$ which satisfies the conditions of Theorem \ref{thm:finitelybounded}.
\end{proof}

\subsubsection{Reducts of the order of the rational numbers}

For the reducts of $(\mQ;\leq)$ the CSP classification in~\cite{tcsps-journal} shows the following, using the notation of~\cite{tcsps-journal,Bodirsky-HDR}.
\begin{theorem}\label{thm:Qclassification}
Let $\mA$ be a reduct of $(\mQ;\leq)$. Then either 
\begin{enumerate}
\item $\Pol(\mA)$ contains one of the operations $\min,\mi,\mx, {\rm ll}$, their duals, or a constant,
\item there is a continuous clone homomorphism from $\Pol(\mA)$ to $\mathbf{1}$.
\end{enumerate}
If $\mA$ has a finite relational language, then $\csp(\mA)$ is tractable in the first case, and NP-complete in the second case.
\end{theorem}

\begin{prop}\label{lem:Qlinearisation}
Theorem \ref{thm:linearisationtheorm} holds for reducts of the order of the rationals $(\mQ;\leq)$.
\end{prop}
\begin{proof}
It suffices to show for a reduct $\relA$ that if $\Pol(\mA)$ contains one of the operations in Theorem \ref{thm:Qclassification} (1), then it satisfies non-trivial linear identities. This is clear if $\pol(\mA)$ contains a constant operation. The operations $\mx$ and $\min$ satisfy the non-trivial linear identities $\mx(x,y) = \mx(y,x)$ and $\min(x,y) = \min(y,x)$, respectively.

For the case when $\mi \in \Pol(\mA)$ we are going to sketch a proof using Lemma \ref{lem:lineq}. 
Let $\beta,\alpha_i,\gamma_i$ be self-embeddings of $(\Q;\leq)$ for $i,j\in\{1,\ldots,5\}$ such that 
\begin{align*}
\beta(x)&<\gamma_1(x)<\gamma_2(x)<\dots<\gamma_5(x) \\ &<\alpha_1(x)<\alpha_2(x)<\dots<\alpha_5(x)<\beta(x+\epsilon)
\end{align*}
 for every $x\in \mQ$ and every $0 < \epsilon\in \mQ$. Then for $1\leq i\leq 5$, the functions
$$
\mi_i(x,y):=\begin{cases}
\alpha_i(x)&\text{if }x<y\\
\beta(x)&\text{if }x=y\\
\gamma_i(y)&\text{if }x>y
\end{cases}
$$
can be written as a composition of $\mi$ with embeddings of $(\mathbb Q;\leq)$, thus they are polymorphisms of $\mA$.
Following the proof of Proposition 10.5.17 in \cite{Bodirsky-HDR}, for each injection $\psi \colon \{1,2,3\} \to \{1,\ldots,5\}$ we can construct $f_{\psi}\in \pol(\mA)$ such that there is an embedding $e \in \overline{\Aut(\Q;\leq)}$ with 
\begin{align*}
f_{\psi}(y,x,x)&=e \circ \mi_{\psi(1)}(y,x)\\
f_{\psi}(x,y,x)&=e \circ \mi_{\psi(2)}(y,x)\\
f_{\psi}(x,x,y)&=e \circ \mi_{\psi(3)}(y,x)\; ;
\end{align*}
hence we found functions satisfying the conditions of Lemma~\ref{lem:lineq}.

We are left with the case where $\pol(\mA)$ contains the binary function ${\rm ll}$. Then by Proposition 10.4.10 in \cite{Bodirsky-HDR},  $\pol(\mA)$ also contains
\begin{align*}
f(x,y,z):=\lex'(&\min(x,y,z), \\ &\max(\min(x,y),\min(x,z),\min(y,z)), \\ &x,y,z),
\end{align*}
where $\lex'$ is a $5$-ary operation that embeds the lexicographical order on $(\mQ;\leq)^5$ into the order $(\mQ;\leq)$. Analogously to the existence of $f$, one can show that $\pol(\mA)$ contains the operation 
\begin{align*}
g(x,y,&z,t,u):=\lex(\min(x,y,z,t,u), \\ &\max(\min(x,y),\min(x,z),\min(x,t),\min(x,u),\\
&\ \min(y,z),\min(y,t),\min(y,u),\min(z,t),\min(z,u)), \\ &x,y,z,t,u),
\end{align*} 
where $\lex$ embeds the lexicographic order on $(\Q;\leq)^7$ into $(\Q;\leq)$. Let $h_1^g,\ldots,h_5^g$ be the nu-minors of $g$. For every finite $F\subseteq \mQ$ and every injective $\psi \colon \{1,2,3\} \to \{1,\ldots,5\}$, it can be easily verified that there is $\alpha_\psi \in \Aut(\mQ;\leq)$ such that on $F$ the identities $\alpha_\psi \circ f(y,x,x)=h_{\psi(1)}^g(x,y)$, $\alpha_\psi  \circ f(x,y,x)=h^g_{\psi(2)}(x,y)$, and $\alpha\circ f(x,x,y)=h^g_{\psi(3)}(x,y)$ hold. Lemma~\ref{lem:endoexistence} then yields a set of functions that satisfies the non-trivial linear identities of Lemma~\ref{lem:lineq}. 
\end{proof}

\subsubsection{Reducts of the random partial order}
In the complexity classification of CSPs for reducts of the random partial order, which we denote by $\mP$, the following dichotomy has been shown~\cite{posetCSP16}; we use the definitions from that article.
\begin{theorem}\label{thm:posetclass}
Let $\mA$ be a reduct of $\mP$. Then one of the following applies.
\begin{itemize}
\item[(1)] $\mA$ is homomorphically equivalent to a reduct of $(\mQ;\leq)$.
\item[(2)] $\pol(\mA)$ contains the binary operation $e_{<}$ or $e_{\leq}$.
\item[(3)] There is a continuous clone homomorphism from $\pol(\mA)$ to ${\mathbf 1}$.
\end{itemize}
If $\mA$ has a finite relational language, then the third case implies that $\csp(\mA)$ is NP-complete, and the second case implies tractability of the CSP.
\end{theorem}

\begin{prop}\label{lem:posetlinearisation}
Theorem \ref{thm:linearisationtheorm} holds for reducts of the random partial order $\mP$.
\end{prop}
\begin{proof}
If a reduct $\mA$ is homomorphically equivalent to a reduct of $(\mQ;\leq)$, then the statement follows from Proposition~\ref{lem:Qlinearisation} and the fact that homomorphic equivalence preserves linear identities (Theorem \ref{thm:h1} in~\cite{wonderland}).
If item~(2) applies, then note that the mappings $(x,y,z)\mapsto e_<(e_<(x,y),z)$ and $(x,y,z)\mapsto e_{\leq}(e_{\leq}(x,y),z)$ are totally symmetric modulo outer embeddings. Since $\mP$ can be described by forbidden substructures of size $3$, the satisfaction of non-trivial linear identities in $\Pol(\relA)$ follows from Theorem \ref{thm:finitelybounded}.
\end{proof}

\subsubsection{Reducts of random graph}
For the random graph $G=(V;E)$, the following dichotomy has been shown~\cite{BodPin-Schaefer-both, BP-reductsRamsey}.
\begin{theorem} \label{thm:CSPclassification}
Let $\mA$ be a reduct of $G$. Then one of the following holds:
\begin{enumerate}
\item $\Pol(\mA)$ contains a constant operation.
\item $\Pol(\mA)$ contains an (at most ternary) injective weak near unanimity function $f(x_1,\ldots,x_k)$ modulo outer embeddings of $G$, i.e., $f$ satisfies identities of the form
\begin{align*}
e_1\circ f(y,x,\ldots,x) &= e_2\circ f(x,y,x,\ldots, x) = \ldots \\ &=e_k\circ f(x,\ldots,x, y), 
\end{align*}
with $e_1,\ldots,e_k \in \overline{\Aut(G)}$.
\item $\Pol(\mA)$ has a continuous homomorphism to ${\mathbf 1}$.
\end{enumerate}
If $\mA$ has a finite relational language, then~(3) implies that $\csp(\mA)$ is NP-complete, and (1) and (2) imply tractability of the CSP.
\end{theorem}

In fact, \cite{BodPin-Schaefer-both} provides a list of concrete weak near unanimity functions modulo outer embeddings that can appear, and Theorem~\ref{thm:finitelybounded} directly applies to a subset of those functions. To obtain non-trivial linear identities for all cases, and moreover simultaneously so, we are going to use  the following variant of Lemma \ref{lem:lineq}.

\begin{lemma} \label{lem:lineq2}
Let $\C$ be a function clone, and suppose there are binary $g_{i,j} \in \C$ for $i, j \in \{1,\ldots, k\}$ such that
\begin{enumerate}
\item for every fixed $j\in\{1,\ldots,k\}$ there is a function $f_j(x_1,\ldots,x_k) \in \C$ whose nu-minors equal $g_{1,j}, \ldots, g_{k,j}$, and
\item for every $\psi\colon \{1,\ldots, k\}\To \{1,\ldots,k\}$ there is a function $f_\psi(x_1,\ldots,x_k) \in \C$ whose nu-minors equal $g_{\psi(1),1}, \ldots, g_{\psi(k),k}$.
\end{enumerate}
Then there is no h1 clone homomorphism from $\C$ to $\mathbf 1$.
\end{lemma}

\begin{proof}
Suppose to the contrary that there exists an h1 clone homomorphism from $\C$ to $\mathbf 1$. If for a fixed $j\in\{1,\ldots,k\}$ the functions $g_{1,j}, \ldots, g_{k,j}$ are mapped to the same projection in $\mathbf 1$, then this contradicts the fact that they are the nu-minors of $f_j(x_1,\ldots,x_k)$. Thus, for every $j\in\{1,\ldots,k\}$ there exists $\psi(j)\in\{1,\ldots,k\}$ such that $g_{\psi(j), j}$ is mapped to the first projection. But then $g_{\psi(1),1}, \ldots, g_{\psi(k),k}$, the nu-minors of $f_\psi$, are all sent to the same projection, a contradiction.
\end{proof}

\begin{prop}\label{lem:randomgraphlinearisation}
Theorem \ref{thm:linearisationtheorm} holds for reducts of the random graph $G=(V;E)$.
\end{prop}

\begin{proof}
If $\Pol(\mA)$ contains a constant operation, then the linear identity $c(x,y) = c(y,x)$ holds for some constant binary $c\in\Pol(\relA)$. So we only have to study the case  where $f(x_1,\ldots,x_k)$ is injective and weak nu modulo outer embeddings of $G$. 

As in Definition~\ref{def:numinors}, denote the nu-minors of $f$ by $h_1^f,\ldots,h_k^f$; we are going to construct the functions $g_{i.j}$, $f_j$, and $f_\psi$ required in Lemma~\ref{lem:lineq2} from these nu-minors. By Lemma~\ref{lem:endoexistence} we only have to prove for every finite $F\subseteq V$ that there are functions $g_{i,j}$, $f_j$, and $f_\psi$ that satisfy the identities in Lemma~\ref{lem:lineq2} on $F$. To this end, we construct a graph $H$ with vertices $(i,j,x,y)$, where $i,j\in \{1,\ldots,k\}$ and $x,y\in F$, and in which two vertices $(i_1,j_1,x_1,y_1)$ and $(i_2,j_2,x_2,y_2)$ are adjacent if and only  if
\begin{itemize}
\item $j_1 = j_2$ and $(h_{i_1}^f(x_1,y_1),h_{i_2}^f(x_2,y_2))\in E$, or
\item $j_1 \neq j_2$ and $(h_{j_1}^f(x_1,y_1),h_{j_2}^f(x_2,y_2))\in E$.
\end{itemize}
By the universality of the random graph we can regard $H$ as a subgraph of $G$. By our construction, for every $j\in\{1,\ldots,k\}$ there exists $\alpha_j\in\Aut(G)$ with $\alpha_j \circ h_i^f(x,y) = (i,j,x,y)$ for all $i\in\{1,\ldots,k\}$ and all $x,y\in F$. Similarly for every $\psi\colon \{1,\ldots,k\}\To \{1,\ldots,k\}$ there exists $\alpha_\psi\in\Aut(G)$ such that $\alpha_\psi \circ h_i^f(x,y) = (\psi(i),i,x,y)$ for all $i\in\{1,\ldots,k\}$ and all $x,y\in F$. It is easy to verify that then $g_{i,j} := \alpha_j \circ h_i^f$, $f_j := \alpha_j \circ f$ and $f_\psi := \alpha_\psi \circ f$ satisfy the equations in Lemma \ref{lem:lineq2} on $F$, and we are done.
\end{proof}


\section{Model-complete Cores: A New Proof}\label{sect:cores}

We are going to give a new and short proof of Theorem~\ref{thm:cores} in the language of monoids. As mentioned in Section~\ref{subsect:cores} of the introduction, our proof will work for \emph{weakly oligomorphic} structures, a generalization of $\omega$-categorical structures~\cite{PechCores}. Those are best defined via their endomorphism monoid.

\begin{defn}\label{defn:wo}
A transformation monoid $\M$ on a countable set $M$ is called \emph{weakly oligomorphic} if for all $n\geq 1$ the equivalence relation $\sim_n$ on $M^n$, given by $a\sim_n b$ if and only if there exist $m,m'\in\M$ such that $a=m(b)$ and $b=m'(a)$, has only finitely many classes. 
A countable structure is called \emph{weakly oligomorphic} if its endomorphism monoid is weakly oligomorphic.  
\end{defn}
Let us remark that weakly oligomorphic monoids have been called oligomorphic in~\cite{PechCores}; this leads, however, to inconsistencies with the corresponding notion for function clones, and so the definition shall henceforth be as stated here.

We first outline the idea behind our proof by recalling the situation for finite structures. When $\relA$ is a structure with finite domain which is not a core, then it has a non-surjective endomorphism. Restricting $\relA$ to the image of that endomorphism, one obtains a homomorphically equivalent structure with smaller domain. After finitely many iterations in this fashion, one obtains a structure which is a core; this structure is the core of $\relA$.

When $\relA$ is infinite and weakly oligomorphic, one could expect the analogous argument to work, where termination of the process after finitely many steps is guaranteed by a compactness (rather than finiteness) argument using weak oligomorphicity. However, this turns out to be insufficient, which is the reason for the argument to become considerably more involved: in addition to the compactness argument, the minimal domain of the model-complete core has to be generic in a sense, which is achieved via a second, Fra\"{i}ss\'{e}-type argument. In particular, contrary to the finite case, in general there is no surjective endomorphism of $\relA$ onto the domain of the model-complete core. It is worth noting that the first step corresponds to the construction of a  structure all of whose endomorphisms are self-embeddings; the second step then constructs a structure where in addition every embedding is elementary, i.e., contained in the closure of its automorphisms. This viewpoint also makes clear why the second step is not present in the finite, as it is automatic. 

We start with compactness. Similarly as in~\cite{Topo-Birk}, we define an equivalence relation on $\M$ in order to obtain a compact object.

\begin{defn}\label{defn:leftequi}
Extending the definition of $\sim_n$ in Definition~\ref{defn:wo}, we denote by $\sim$ the equivalence relation on $\M\subseteq M^M$ defined by $f\sim g$ if for all $n\geq 1$ and all $x\in M^n$ we have $f(x)\sim_n g(x)$.
\end{defn}

The standard K\H{o}nig's lemma argument proving the following lemma has been executed in~\cite{BodJunker} for a finer equivalence relation and monoids containing an oligomorphic permutation group, then again in~\cite{Topo-Birk} for the case of oligomorphic function clones, and has once again been presented, perhaps more conceptually, in the most general context in~\cite{canonical}  -- the proof here would be identical, so we omit it.

\begin{lem}\label{lem:compact}
If $\M$ is a topologically closed weakly oligomorphic transformation monoid, then the factor space $\M/\sim$ is compact.
\end{lem}

\begin{lem}\label{lem:ideal}
If $\M$ is a topologically closed weakly oligomorphic transformation monoid, then $\M$ contains a minimal non-empty topologically closed left ideal.
\end{lem}
\begin{proof}
Consider the set $S$ of all non-empty topologically closed subsets $I'$ of $\M/\sim$ with the property that whenever $[f]_\sim\in I'$ and $m\in \M$, then $[m\circ f]_\sim\in I'$. Then by compactness arbitrary descending chains in $S$  have a non-empty intersection in $S$. Hence, by Zorn's lemma, $S$ contains a minimal element, the preimage of which under the factor mapping from $\M$ to $\M/\sim$ is topologically closed, left-invariant, and minimal with this property.
\end{proof}



In a sense, any function in a minimal non-empty (topologically) closed left ideal of $\M$ as guaranteed by Lemma~\ref{lem:ideal} can be considered to have minimal range, analogous to the finite case described above. We now argue that this minimal range gives rise to a generic structure. 
Our Fra\"{i}ss\'{e}-type argument (cf.~\cite{OriginalFraisse, Fraisse}) can be performed either by introducing a suitable relational language, or via a more general category-theoretic approach~\cite{kubis}. For brevity we choose the former and define a relational structure $\mathbb{M}$ on the domain $M$ by introducing, for each equivalence class of each relation $\sim_n$,  an $n$-ary relation equal to this class. Subsets of $M$ will be regarded as induced substructures of $\relstr{M}$.
Note that when $I\subseteq \M$ is a minimal closed left ideal, $g\in I$, and $F, F'$ are finite subsets in the range of $g$, then $F$ embeds into $F'$ if and only if there exists $m\in\M$ such that $m[F]\subseteq F'$. 

\begin{lem}\label{lem:fraisse}
Let $I$ be a minimal non-empty closed left ideal of a closed transformation monoid $\M$ and let $g \in I$. Then 
$$
\{g[F] \; | \; F\subseteq M\text{ finite}\}
$$ is a Fra\"{i}ss\'{e} category under embeddings as above.
\end{lem}
\begin{proof}
We check the amalgamation property. Consider structures $g[F], g[F_1], g[F_2]$ in the above set and embeddings $m_1\colon g[F] \to g[F_1]$, $m_2\colon g[F] \to g[F_2]$. Since $I$ is a minimal non-empty closed left ideal, there exist $m_1', m_2' \in\M$ such that $m_1'\circ m_1 \circ g \rest_{F_1}=g \rest_{F_1}$ and $m_2'\circ m_2 \circ g \rest_{F_2}=g \rest_{F_2}$. Hence, $g[F_1\cup F_2]$ is an amalgam.
\end{proof}

The Fra\"{i}ss\'{e} limit $\mathbb X$ of the category in Lemma \ref{lem:fraisse} yields the  model-complete core of $\M$ as follows.

\begin{lem}\label{lem:shinkmonoid}
Let $\mathbb X$ be the Fra\"{i}ss\'{e} limit of any Fra\"{i}ss\'{e} category as in Lemma~\ref{lem:fraisse}. Then the monoid
\begin{align*}
\cM:=\{ f\in X^X\; |\; &\forall F\subseteq X \text{ finite }\\ 
&f\rest_F \text{ is an embedding from } F \text{ onto } f[F]\; \}
\end{align*}
is a model-complete core, i.e., has dense invertibles. Moreover, if $\M$ is weakly oligomorphic, then $\cM$ is oligomorphic.
\end{lem}
\begin{proof}
By the homogeneity of $\mathbb X$, the monoid $\cM$ is equal to the closure of $\Aut(\mathbb X)$ in $X^X$. Hence, $\cM$ is closed and a model-complete core. If $\M$ is weakly oligomorphic, then $\cM$ is oligomorphic since all equivalence relations $\sim_n$ have only finitely many classes.
\end{proof}

We can now derive Theorem~\ref{thm:cores} in its more general form for weakly oligomorphic structures.

\begin{thm}\label{thm:cores:wo}
Every weakly oligomorphic structure $\relA$ is homomorphically equivalent to a model-complete core $\relB$. Moreover, $\relB$ is $\omega$-categorical and unique up to isomorphism.
\end{thm}
\begin{proof}
Set $\M:=\End(\relA)$, and let $I$ be a minimal non-empty closed left ideal. 
For an arbitrarily chosen $g \in I$, let $\mathbb Y$ be the induced substructure of $\mathbb M$ on the range $Y$ of $g$.
Further let $\mathbb X$ be the Fra\"{i}ss\'{e} limit given by $g$ and Lemma~\ref{lem:fraisse} and let $\cM$ be as in Lemma~\ref{lem:shinkmonoid}.
Because $\M$ is weakly oligomorphic, both $\mathbb X$ and $\mathbb Y$ are $\omega$-categorical; moreover, by definition $\mathbb X$ and $\mathbb Y$ have the same finite substructures. It is well-known that whenever two $\omega$-categorical structures have the same finitely generated substructures, then they embed into each other (cf.~\cite{Oligo}), and whence there is an embedding of $\mathbb X$ into $\mathbb Y$ and vice-versa. We may thus henceforth assume that $\mathbb X$ is a substructure of $\mathbb Y$.

We set $\relB$ to be the induced substructure of $\relA$ on $X$ (so $B=X$). Then $\relA$ and $\relB$ are homomorphically equivalent: $\relB$ is a substructure of $\relA$, and composing $g$ with an embedding of $\mathbb Y$ into $\mathbb X$ yields a homomorphism from $\relA$ into $\relB$. 

We next show that $\End(\relB)=\cM$. The only non-trivial inclusion being $\End(\relB)\subseteq\cM$, let $e\in \End(\relB)$, and let $x$ be a finite tuple of elements in $B=X$; we find an element of $\cM$ which agrees with $e$ on $x$. By the definitinon of $\cM$ and since it is a model-complete core, it suffices to find $m\in \M$ such that $m\circ e(x)=x$. To this end, note that we can write $x$ as $g(y)$, for some tuple $y$ in $A$. By the minimality of $I$, there exists $m'\in\M$ such that $m'\circ g(x)=x$. Again by minimality, and since $e\circ m' \circ g\in\M$, there exists $m\in\M$ such that $m\circ e\circ m'\circ g(g(y))=g(y)$, so 
 $$
 m\circ e(x)=m\circ e\circ m' \circ g(x)= m\circ e\circ m' \circ g(g(y))=g(y)=x\; ,
 $$
 proving our claim.

By Lemma~\ref{lem:shinkmonoid}, $\End(\relB)=\cM$ is oligomorphic and a model-complete core. Hence, $\relB$ is $\omega$-categorical. Its uniqueness follows easily from the definitions, as in previous well-known proofs.
\end{proof}

Finally, we connect the concepts of \emph{model-complete cores} and \emph{reflections}. Let $\C$ be a function clone on a set $C$, let $D$ be a set, and let $u\colon C\To D$ and $v\colon D\To C$ be functions. The \emph{reflection} of $\C$ by $u,v$ is the set
$$
\{u(t(v(x_1),\ldots,v(x_n))\; |\; t\in\C\}.
$$
The reflection of a transformation monoid is defined similarly~\cite{wonderland}. The following can be derived directly from Lemma~\ref{lem:shinkmonoid}, without proving Theorem~\ref{thm:cores:wo}.

\begin{prop}\label{prop:reflections}
Let $\M$ be a weakly oligomorphic closed transformation monoid. Then it has a reflection contained in an oligomorphic closed model-complete core $\cM$, which in turn has a reflection contained in $\M$.
\end{prop}
\begin{proof}
Let $I, g,\mathbb Y, \mathbb X, \cM$ be as above. We set $u\colon A\To X$ to be $g$ composed with any embedding from $\mathbb Y$ into $\mathbb X$, and $v\colon X\To A$ to be any embedding from $\mathbb X$ into $\mathbb Y$. The reflection $\{u\circ m\circ v\;|\; m\in\M\}$ is contained in $\cM$. Conversely, the reflection $\{ v\circ m\circ u \;|\; m\in\cM\}$ is contained in $\M$.
\end{proof}

If we are not interested in obtaining Theorem~\ref{thm:cores}, but only in its utility for CSPs, then we do not need to use Theorem~\ref{thm:cores:wo}, but can directly apply Proposition~\ref{prop:reflections}. A function clone is called \emph{weakly oligomorphic} if the monoid of its unary functions is weakly oligomorphic.

\begin{cor}\label{cor:reflectionsclones}
Every weakly oligomorphic closed function clone $\C$ has a reflection contained in an oligomorphic closed model-complete core $\cC$, which in turn has a reflection contained in $\C$. In particular, the CSP of any weakly oligomorphic structure is polynomial-time equivalent to the CSP of an oligomorphic model-complete core.
\end{cor}
\begin{proof}
The proof of the first statement is identical to that of Proposition~\ref{prop:reflections}. The second statement then follows from~\cite[Proposition~4.6]{wonderland}.
\end{proof}

    \bibliographystyle{alpha}
\bibliography{global.bib,CSPbib.bib}

\def\cprime{$'$} \def\cprime{$'$}
\begin{thebibliography}{KPT05}

\bibitem[BEKP]{BodirskyEvansKompatscherPinsker}
Manuel Bodirsky, David Evans, Michael Kompatscher, and Michael Pinsker.
\newblock A counterexample to the reconstruction of $\omega$-categorical
  structures from their endomorphism monoids.
\newblock {\em Israel Journal of Mathematics}.
\newblock To appear. Preprint arXiv:1510.00356.

\bibitem[Ber11]{Berg11}
Clifford Bergman.
\newblock {\em Universal Algebra: Fundamentals and Selected Topics}.
\newblock Pure and Applied Mathematics. Taylor and Francis, 2011.

\bibitem[BG08]{BodirskyGrohe}
Manuel Bodirsky and Martin Grohe.
\newblock Non-dichotomies in constraint satisfaction complexity.
\newblock In Luca Aceto, Ivan Damgard, Leslie~Ann Goldberg, Magn\'us~M.
  Halld\'orsson, Anna Ing\'olfsd\'ottir, and Igor Walukiewicz, editors, {\em
  Proceedings of the International Colloquium on Automata, Languages and
  Programming (ICALP)}, Lecture Notes in Computer Science, pages 184 --196.
  Springer Verlag, July 2008.

\bibitem[BJ11]{BodJunker}
Manuel Bodirsky and Markus Junker.
\newblock $\aleph_0$-categorical structures: interpretations and endomorphisms.
\newblock {\em Algebra Universalis}, 64(3-4):403--417, 2011.

\bibitem[BK08]{ecsps}
Manuel Bodirsky and Jan K\'ara.
\newblock The complexity of equality constraint languages.
\newblock {\em Theory of Computing Systems}, 3(2):136--158, 2008.
\newblock A conference version appeared in the proceedings of Computer Science
  Russia {(CSR'06)}.

\bibitem[BK09]{tcsps-journal}
Manuel Bodirsky and Jan K\'ara.
\newblock The complexity of temporal constraint satisfaction problems.
\newblock {\em Journal of the ACM}, 57(2):1--41, 2009.
\newblock An extended abstract appeared in the Proceedings of the Symposium on
  Theory of Computing (STOC).

\bibitem[Bod07]{Cores-journal}
Manuel Bodirsky.
\newblock Cores of countably categorical structures.
\newblock {\em Logical Methods in Computer Science}, 3(1):1--16, 2007.

\bibitem[Bod12]{Bodirsky-HDR}
Manuel Bodirsky.
\newblock Complexity classification in infinite-domain constraint satisfaction.
\newblock M\'emoire d'habilitation \`a diriger des recherches, Universit\'{e}
  Diderot -- Paris 7. Available at arXiv:1201.0856, 2012.

\bibitem[BOP]{wonderland}
Libor Barto, Jakub Opr\v{s}al, and Michael Pinsker.
\newblock The wonderland of reflections.
\newblock {\em Israel Journal of Mathematics}.
\newblock To appear. Preprint arXiv:1510.04521.

\bibitem[BP11]{BP-reductsRamsey}
Manuel Bodirsky and Michael Pinsker.
\newblock Reducts of {R}amsey structures.
\newblock {\em AMS Contemporary Mathematics, vol. 558 (Model Theoretic Methods
  in Finite Combinatorics)}, pages 489--519, 2011.

\bibitem[BP15a]{BodPin-Schaefer-both}
Manuel Bodirsky and Michael Pinsker.
\newblock Schaefer's theorem for graphs.
\newblock {\em Journal of the ACM}, 62(3):52 pages (article number 19), 2015.
\newblock A conference version appeared in the Proceedings of STOC 2011, pages
  655--664.

\bibitem[BP15b]{Topo-Birk}
Manuel Bodirsky and Michael Pinsker.
\newblock Topological {B}irkhoff.
\newblock {\em Transactions of the American Mathematical Society},
  367:2527--2549, 2015.

\bibitem[BP16a]{BartoPinskerDichotomy}
Libor Barto and Michael Pinsker.
\newblock The algebraic dichotomy conjecture for infinite domain constraint
  satisfaction problems.
\newblock In {\em Proceedings of LICS'16}, pages 615--622, 2016.
\newblock Preprint arXiv:1602.04353.

\bibitem[BP16b]{canonical}
Manuel Bodirsky and Michael Pinsker.
\newblock Canonical functions: a new proof via topological dynamics.
\newblock Preprint arXiv:1610.09660, 2016.

\bibitem[BPP]{BPP-projective-homomorphisms}
Manuel Bodirsky, Michael Pinsker, and Andr\'{a}s Pongr\'acz.
\newblock Projective clone homomorphisms.
\newblock {\em Journal of Symbolic Logic}.
\newblock To appear. Preprint arXiv:1409.4601.

\bibitem[BPP17]{Reconstruction}
Manuel Bodirsky, Michael Pinsker, and Andr\'{a}s Pongr\'acz.
\newblock Reconstructing the topology of clones.
\newblock {\em Transactions of the American Mathematical Society},
  369:3707--3740, 2017.

\bibitem[BS81]{BS81}
Stanley~N. Burris and H.~P. Sankappanavar.
\newblock {\em A course in universal algebra}, volume~78 of {\em Graduate Texts
  in Mathematics}.
\newblock Springer-Verlag, New York, 1981.

\bibitem[Cam90]{Oligo}
Peter~J. Cameron.
\newblock {\em Oligomorphic permutation groups}.
\newblock Cambridge University Press, Cambridge, 1990.

\bibitem[Fra54]{OriginalFraisse}
Roland Fra{\"\i}ss\'e.
\newblock Sur l'extension aux relations de quelques propri\'et\'es des ordres.
\newblock {\em Annales Scientifiques de l'\'Ecole Normale Sup\'erieure},
  71:363--388, 1954.

\bibitem[Fra86]{Fraisse}
Roland Fra{\"\i}ss\'e.
\newblock {\em Theory of Relations}.
\newblock Elsevier Science Ltd, North-Holland, 1986.

\bibitem[FV99]{FederVardi}
Tom\'as Feder and Moshe~Y. Vardi.
\newblock The computational structure of monotone monadic {SNP} and constraint
  satisfaction: {a} study through {D}atalog and group theory.
\newblock {\em {SIAM} Journal on Computing}, 28:57--104, 1999.

\bibitem[GP]{uniformbirkhoff}
Mai Gehrke and Michael Pinsker.
\newblock Uniform {B}irkhoff.
\newblock {\em Journal of Pure and Applied Algebra}.
\newblock To appear. Preprint available from the second author's website.

\bibitem[Hod97]{Hodges}
Wilfrid Hodges.
\newblock {\em A shorter model theory}.
\newblock Cambridge University Press, Cambridge, 1997.

\bibitem[KP17]{posetCSP16}
Michael Kompatscher and Trung~Van Pham.
\newblock {A Complexity Dichotomy for Poset Constraint Satisfaction}.
\newblock In {\em 34th Symposium on Theoretical Aspects of Computer Science
  (STACS 2017)}, volume~66, pages 47:1--47:12, 2017.

\bibitem[KPT05]{Topo-Dynamics}
Alexander Kechris, Vladimir Pestov, and Stevo Todorcevic.
\newblock Fra\"{i}ss\'e limits, {R}amsey theory, and topological dynamics of
  automorphism groups.
\newblock {\em Geometric and Functional Analysis}, 15(1):106--189, 2005.

\bibitem[Kub14]{kubis}
Wieslaw Kubi\'{s}.
\newblock Fra\"{i}ss\'{e} sequences: category-theoretic approach to universal
  homogeneous structures.
\newblock {\em Annals of Pure and Applied Logic}, 165:1755 -- 1811, 2014.

\bibitem[Mac11]{MacphersonSurvey}
Dugald Macpherson.
\newblock A survey of homogeneous structures.
\newblock {\em Discrete Mathematics}, 311(15):1599--1634, 2011.

\bibitem[PP16]{PechCores}
Christian Pech and Maja Pech.
\newblock Towards a {Ryll}-{Nardzewski}-type theorem for weakly oligomorphic
  structures.
\newblock {\em Mathematical Logic Quaterly}, 62(1--2):25--34, 2016.

\bibitem[Sch15]{schneider}
Friedrich~Martin Schneider.
\newblock A uniform {B}irkhoff theorem.
\newblock Preprint arXiv:1510.03166, 2015.

\bibitem[Tho91]{RandomReducts}
Simon Thomas.
\newblock Reducts of the random graph.
\newblock {\em Journal of Symbolic Logic}, 56(1):176--181, 1991.

\end{thebibliography}

\end{document}